\documentclass[conference,compsoc]{IEEEtran}
\ifCLASSOPTIONcompsoc
  \usepackage[nocompress]{cite}
\else
  \usepackage{cite}
\fi
\ifCLASSINFOpdf
\else
\fi
\usepackage{amsthm} 
\usepackage{multirow}
\usepackage{booktabs} 
\usepackage{pifont}   
\usepackage{tabularx}
\usepackage{makecell}
\usepackage{hyperref}

\usepackage{amsmath}
\usepackage{algorithm, algpseudocode}
\usepackage{setspace}
\usepackage{array}
\usepackage{url}
\usepackage{xspace}
\usepackage{xcolor}

\usepackage{amssymb}
\usepackage{threeparttable}

\usepackage{graphicx}
\usepackage{subcaption}

\usepackage{enumitem}

\usepackage{mathtools}
\DeclarePairedDelimiter\ceil{\lceil}{\rceil}

\newtheorem{definition}{Definition}
\newtheorem{lemma}{Lemma}

\newcommand{\sys}{\textsc{CoBRA}\xspace}

\newcommand{\com}[1]{}

 \newcommand{\ray}[1]{}
  \newcommand{\lef}[1]{}
   \newcommand{\za}[1]{}
    \newcommand{\chris}[1]{}

\com{
  \newcommand{\ray}[1]{{ [\color{brown} Ray: #1]}}
 \newcommand{\lef}[1]{{ \color{blue} [Lefteris: #1]}}
 \newcommand{\za}[1]{{\color{red} [Zeta: #1]}}
  \newcommand{\chris}[1]{{ [\color{purple} Christos: #1]}}
} 

\algblockdefx[<block>]{Upon}{EndUpon}[2][]{\textbf{on} #2 \textbf{do}}{}
\algtext*{EndIf}
\algtext*{EndFor}
\algtext*{EndProcedure}

\newtheorem{theorem}{Theorem}[section]
\usepackage{thmtools} 

\newtheorem{corollary}[theorem]{Corollary}

\numberwithin{theorem}{section}  

\begin{document}
%
\title{CoBRA: A Universal Strategyproof Confirmation Protocol for Quorum-based Proof-of-Stake Blockchains}

\com{
\author{Zeta Avarikioti}
\affiliation{%
 \institution{TU Wien \& Common Prefix}
 \city{Vienna}
  \country{Austria}}

\author{Eleftherios Kokoris Kogias}
\affiliation{%
  \institution{Mysten Labs}
  \city{Athens}
  \country{Greece}
}

\author{Ray Neiheiser}
\affiliation{%
  \institution{ISTA}
  \city{Vienna}
  \country{Austria}}

\author{Christos Stefo}
\affiliation{%
  \institution{TU Wien}
 \city{Vienna}
  \country{Austria}}
}

\author{
\IEEEauthorblockN{
Zeta Avarikioti\IEEEauthorrefmark{1},
Eleftherios Kokoris Kogias\IEEEauthorrefmark{1}\IEEEauthorrefmark{2},
Ray Neiheiser\IEEEauthorrefmark{3},
Christos Stefo\IEEEauthorrefmark{4}
}

\IEEEauthorblockA{\IEEEauthorrefmark{1}TU Wien \& Common Prefix}
\IEEEauthorblockA{\IEEEauthorrefmark{2}Mysten Labs}
\IEEEauthorblockA{\IEEEauthorrefmark{3}ISTA}
\IEEEauthorblockA{\IEEEauthorrefmark{4}TU Wien}
}

\maketitle



%

\begin{abstract} 
The integrity of many Proof-of-Stake (PoS) payment systems relies on quorum-based State Machine Replication (SMR) protocols  mitigating diverse adversarial attacks. While traditional analyses assume a purely Byzantine threat model, practical security depends on resilience against both arbitrary faults and strategic, profit-driven actors. In this work, we study this challenge by analyzing SMR protocols under a hybrid threat model comprising honest, Byzantine, and rational validators. We first establish the fundamental limitations of traditional consensus mechanisms, proving two impossibility results: (1) in partially synchronous networks, no quorum-based protocol can achieve SMR when rational and Byzantine validators collectively exceed $1/3$ of the participants; and (2) even under synchronous network assumptions, SMR remains unattainable if this coalition comprises more than $2/3$ of the validator set.

Motivated by these boundaries, we develop a protocol that achieves SMR in the presence of up to $1/3$ Byzantine and $1/3$ rational validators. Our approach relies on two complementary mechanisms to circumvent our impossibility results. First, we introduce a protocol constraint that enforces an upper bound on the transaction volume finalized within any time window $\Delta$, and we prove that such a bound is a necessary condition for security. Second, we define the \emph{strongest chain rule}, a finality gadget that enables instant transaction confirmation when a supermajority of participants provably supports the execution. We validate the feasibility of this design through an empirical analysis of the Ethereum and Cosmos networks, demonstrating that validator participation exceeds the required $5/6$ threshold in over $99\%$ of observed blocks.

Finally, we explore adversarial settings beyond the classical Byzantine threshold. CoBRA supports recovery from consistency violations even when Byzantine validators control less than $2/3$ of the stake, while providing an additional client-level economic security guarantee: every transaction accepted by clients, including those on conflicting branches, is accounted for during recovery, enabling full reimbursement of provable losses without requiring client participation and without minting new coins.

\end{abstract}

\newif\ifsubmissionversion
\submissionversionfalse  

\begin{table*}[t]
\centering
\footnotesize
\setlength{\tabcolsep}{4pt}
\renewcommand{\arraystretch}{1.1}

\begin{threeparttable}

\begin{minipage}{0.50\textwidth}
\centering
\begin{tabular}{
>{\raggedright\arraybackslash}p{3.5cm}
|>{\centering\arraybackslash}p{2.5cm}
|>{\centering\arraybackslash}p{2.8cm}
|>{\centering\arraybackslash}p{1.8cm}
|>{\centering\arraybackslash}p{1.8cm}
}
\toprule
\textbf{Protocol} &
\textbf{Safety / Liveness Resilience} &
\textbf{Recovery Resilience} &
\textbf{Bounded Rollback} &
\textbf{Client Reimbursement} \\
\midrule

Classic BFT \cite{hotstuff,buchman2018latest,pbft} &
$n/3$ Byzantine &
-- &
-- &
-- \\
\midrule

Accountable + Recoverable Partial Synchrony \cite{gong2025recover,zlb} &
$n/3$ Byzantine &
$5n/9$ Byzantine &
\ding{55} &
\ding{55} \\
\midrule

Accountable + Recoverable Synchrony \cite{lewis2025beyond} &
$n/3$ Byzantine &
$2n/3$ Byzantine &
\checkmark &
\ding{55} \\
\midrule

\textbf{\sys (ours.)\textsuperscript{*}} &
$n/3$ Byzantine $+\, n/3$ rational &
$2n/3$ Byzantine &
\checkmark &
\checkmark \\
\bottomrule
\end{tabular}
\end{minipage}
\hfill





\caption{\footnotesize Comparison of classical BFT, accountable/recoverable consensus, and \sys. Bounded rollback means that the depth of the chain reversed after recovery is limited; client reimbursement means that every honest client is fully reimbursed, even for transactions accepted during safety violations.}
\label{tab:protocol-comparison}
\vspace{-0.4cm}
\end{threeparttable}
\end{table*}


\section{Introduction}
State Machine Replication (SMR) protocols are the backbone of secure distributed ledgers, ensuring agreement on an ordered sequence of transactions among distributed participants.
Traditional SMR protocols model participants as either \emph{honest}, strictly following the protocol, or \emph{Byzantine}, deviating arbitrarily.
Quorum-based Byzantine Fault Tolerant (BFT) protocols, such as PBFT~\cite{pbft}, HotStuff~\cite{yin2018hotstuff}, and Tendermint~\cite{buchman2018latest}, operate under partial synchrony and guarantee safety and liveness provided that less than $1/3$ of the participants are Byzantine.
Once this threshold is exceeded, safety violations are unavoidable.

While the Byzantine model is deliberately pessimistic, the complementary assumption that a $2/3$ supermajority of participants always behaves honestly is particularly strong.
This is evident in modern Proof-of-Stake (PoS) blockchains, where participation is economically motivated.
Validators stake assets to earn rewards from block production, transaction fees, and external financial positions, collectively securing hundreds of billions of dollars in value.
In this setting, participants are not altruistic but profit-maximizing agents.
Deviations from the protocol may thus be \emph{strategic} rather than arbitrary, arising whenever doing so increases expected utility.
As a result, violations of the honest supermajority assumption may stem not only from malicious behavior, but also from economically rational behavior.


This observation raises a fundamental question for the design of secure distributed ledgers:
What assurances can we provide when fewer than $2/3$ of the participants behave honestly?
In particular, can we design protocols that both (i) incentivize rational, profit-maximizing validators to follow the protocol, and (ii) recover from safety violations without imposing irreversible economic losses on clients, even when more than $1/3$ of the participants behave Byzantine?

\subsection{Related Work}

A substantial body of work has studied consensus under violations of the classic $2/3$ honest supermajority assumption.
Broadly, this literature falls into two main categories.

\vspace{3pt}
\noindent\textbf{Rational and hybrid fault models.}
A first line of work studies consensus under hybrid fault models distinguishing honest, rational, and Byzantine participants. 
BAR~\cite{bar} formalizes this setting, and subsequent works extend it to incentive-compatible and equilibrium-based consensus, ensuring honest behavior is a Nash equilibrium under appropriate rewards and penalties~\cite{amoussou2019rationals,mcmenamin2021achieving,lev2019fairledger}. These approaches typically assume rational participants act independently, ignoring collusion. This assumption is problematic in open, financially motivated environments such as Proof-of-Stake blockchains, where validators often coordinate in pools, making deviations like double-spending profitable when executed jointly. 

To capture collusions, Abraham et al.~\cite{abraham2006distributed,abraham2008lower} formalize solution concepts for strategic collusion between rational and Byzantine participants in secret sharing and multi-party computation. 
Building on this, TRAP~\cite{trap} introduces a baiting game to hold misbehaving nodes accountable, solving one-shot consensus in a partially synchronous model for \( n > \max \left( \frac{3}{2}k + 3f, 2(k + f) \right) \) where $k$ in the number of rational and $f$ the number of Byzantine participants. Their analysis assumes an upper bound on the maximum value that can be double-spent against clients.
While one-shot consensus allows bounding the maximum double-spend, this {does not extend to multi-shot SMR}, where repeated strategic behavior can accumulate unbounded economic damage.

\vspace{3pt}
\noindent\textbf{Accountability and recovery.}
A second line of work studies consensus protocols that tolerate violations of the classic $n/3$ Byzantine threshold by weakening safety guarantees and providing \emph{accountability}~\cite{sheng2021bft, polygraph,abc}  and \emph{recovery mechanisms}~\cite{lewis2025beyond,zlb,gong2025recover}.

Accountable protocols focus on detecting equivocation and identifying misbehaving validators. However, accountability alone is not sufficient to deter safety attacks. If the profit validators can gain from double-spending exceeds their staked collateral, they may still choose to behave maliciously, even if their stake is eventually slashed. Moreover, accountability itself does not restore safety or protect clients from economic losses incurred during safety violations.

Recoverable consensus protocols address this limitation by incorporating a mechanism that enables the system to internally restore a consistent state after safety violations, even when more than $1/3$ of the participants behave maliciously~\cite{gong2025recover, lewis2025beyond, zlb, sridhar2023bettersafesorryrecovering}.
However, these protocols either assume communication between clients to detect conflicting certificates~\cite{sridhar2023bettersafesorryrecovering} or provide no economic guarantees to clients~\cite{gong2025recover, lewis2025beyond, zlb}: transactions accepted during safety violations may be reverted without compensation.

As a result, existing work does not resolve our fundamental question. There is no SMR protocol that (i) incentivizes profit-maximizing participants to follow the protocol and (ii) provides economic security to clients when the classic $1/3$ Byzantine threshold is exceeded.



\subsection{Our Contributions}
This paper presents \emph{CoBRA} (Correct over Byzantine-Rational Adversary), which closes this gap. It is the first generic transformation for quorum-based SMR that provides optimal rational resilience when the Byzantine parties are less than 1/3 and economic protection for clients when they are more than 1/3 and less than 2/3. 
Table~\ref{tab:protocol-comparison} compares CoBRA to prior work, highlighting trade-offs in fault tolerance, recovery guarantees, and client-level economic protection.

\vspace{3pt}
\noindent\textbf{Impossibility results.}
We establish fundamental limits for quorum-based SMR under a hybrid fault model with honest, Byzantine, and rational validators, and non-interactive clients:
\begin{itemize}[nosep,leftmargin=*]
    \item Under \emph{partial synchrony}, we show that quorum-based SMR is impossible once Byzantine and rational validators together control at least $1/3$ of the validator set.  This bound is tight and generalizes the classic Byzantine threshold to economically motivated deviations.
    \item Under \emph{synchrony}, we further prove that SMR remains impossible when this coalition exceeds $2/3$. 
\end{itemize}
\za{To include feasibility results about the design space?}

\vspace{3pt}
\noindent\textbf{Protocol transformation.}
We introduce CoBRA, a generic transformation applicable to existing quorum-based PoS protocols, yielding three operational regimes:
\begin{itemize}[nosep,leftmargin=*]
    \item \emph{Standard regime:} In partial synchrony, \sys preserves the safety and liveness guarantees of the underlying SMR protocol against up to $1/3$ Byzantine validators.
    \item \emph{Rational-resilient regime:} Assuming a (possibly large) known synchronous bound $\Delta^*$ on message delays, \sys ensures safety and liveness even when, in addition to $1/3$ Byzantine validators, up to another $1/3$ of the participants are rational and may engage in strategic collusion. This is achieved by modifying only the finalization rule of the SMR protocol.
    \item \emph{Recoverable economic-security regime:} Assuming a synchronous bound $\Delta^*$, when the fraction of Byzantine validators exceeds the standard $1/3$ threshold but remains below $2/3$ of the total validator set, \sys recovers from consistency violations and guarantees full client reimbursement. Recovery builds upon the set-agreement procedure of~\cite{lewis2025beyond}, while our novel reimbursement guarantee follows from \sys's finalization rule, which bounds the value that can be double-spent during an equivocation by the stake of provably misbehaving validators. 
    \sys provides these guarantees without minting new coins and without client participation in the recovery procedure.
\end{itemize}

\section{Preliminaries}

\subsection{Model and Assumptions}




In this section, we outline the model and assumptions used throughout the paper. 

\vspace{3pt}
\noindent\textbf{Proof-of-stake blockchain.}
A PoS blockchain implements a distributed ledger protocol~\cite{bitcoinproperties} (i.e., a state machine replication protocol), where participants are selected to create and validate new blocks based on the amount of coins (cryptocurrency) they hold and lock up as stake. 
The blockchain serves as a payment system in which participants hold and transfer fungible assets, referred to as coins. We also discuss how to extend this model to general-purpose smart contract platforms in Section~\ref{sec:discussion}.

Each party is associated with a private-public key pair which forms its \emph{client} identity in the system. 
The parties that actively participate in maintaining the state of the system (i.e., run the SMR) are known as \emph{validators}. To do so, they lock coins (i.e., stake) and process transactions, providing clients an ordered sequence of confirmed transactions, referred to as the \emph{transaction ledger}. We assume the existence of a genesis block $G$, which contains the initial stake distribution and the identities of the \emph{validators} participating in the SMR protocol. For the remainder of this paper, we assume a static validator setting, where the stake distribution is known in advance to all participants. In Section~\ref{sec:discussion}, we discuss how to extend our system to a dynamic validator setting.

\vspace{3pt}
\noindent\textbf{Validators.}
For simplicity, we assume a uniform model where all validators lock the same stake equal to $D$ coins and have thus the same ``voting power''. 
If an entity locks some multiple $\alpha\cdot D$, they may control up to $\alpha$ individual validators. 
Thus, the SMR protocol involves a fixed set of validators, \(N\), which are known to all participants. Let \(n = |N|\) denote the total number of validators. We assume a public key infrastructure (PKI), where each validator is equipped with a public-private key pair for signing messages, and their public keys are available to all parties.

\vspace{3pt}
\noindent\textbf{Clients.}
Unlike validators, the number of active clients is not known to all parties. Clients have public keys, which they can use to issue authenticated transactions. We assume clients are \emph{silent}~\cite{sridhar2024consensus}, meaning they neither actively interact with nor observe the execution of the SMR protocol among validators. Instead, clients query a subset of validators to obtain confirmation of any transaction (i.e., the transaction ledger) from their responses.


\vspace{3pt}
\noindent\textbf{Cryptography.}
We assume that all participants are computationally bounded, i.e., they are modeled as probabilistic polynomial-time (PPT) interactive Turing machines. We further assume the existence of cryptographically secure communication channels, hash functions, and digital signatures.

\vspace{3pt}
\noindent\textbf{Time.}
Time progresses in discrete logical units indexed as \(u = 0, 1, 2, \dots\). Each participant has their own local clock. We assume that the maximum clock offset between any two parties is bounded, as any such offset can be absorbed into the network delay bound.  

\vspace{3pt}
\noindent\textbf{Network.}
We assume all participants are connected through perfect point-to-point channels. We consider two communication models: \emph{partial synchrony} and \emph{synchrony}. In the former we show impossibility results; in the latter, both impossibilities and main results.
In \emph{partial synchrony:} there is an unknown Global Stabilization Time ($GST$) after which messages arrive within a known bound $\Delta$; i.e., a message sent at time $t$ arrives by $t' \leq \max\{GST, t\} + \Delta$.
In \emph{synchrony:} every message sent at time $t$ arrives by $t + \Delta^*$, for a known bound $\Delta^*$.

\vspace{3pt}
\noindent\textbf{Adversarial model.}
We consider a mixed model where \emph{validators are honest, Byzantine, or rational}. Honest (or \emph{correct}) validators always follow the protocol. Byzantine validators may deviate arbitrarily from the protocol. We assume a Byzantine adversary does not have any liquid coins, i.e.,  all their coins are staked to be able to maximize their influence. 
Rational validators, on the other hand, will deviate from the protocol specification if and only if doing so maximizes their utility. We define a party's utility as its monetary profit, calculated as the rewards from participation minus potential penalties (e.g., slashing). 

We denote the number of Byzantine and rational validators by $f^*$ and $k$, respectively. In Section~\ref{sec:impossibilitiesSummary}, we derive the feasible bounds for $f^*$ and $k$. In our solution, we consider $n = 3f + 1$, where $f$ represents the maximum number of Byzantine validators tolerated by standard BFT protocols, e.g., \cite{hotstuff,buchman2018latest,pbft}. In Section~\ref{sec:protocol}, we specify that $f^* \leq f$ and $k \leq 2f - f^*$. Then, in Section~\ref{sec:recovery}, we extend our solution to Byzantine failures beyond the classic threshold, in particular, $f+1 \leq f^* < 2n/3$. 

We make the following assumptions on rational validators:
(a)~We assume that participation rewards exceed both the opportunity cost and the cost of operating a validator. Thus, these costs are treated as negligible, simplifying the analysis by offsetting i) the constant operational expenses (e.g., bandwidth, CPU, storage), ii) voting for a block proposed by another validator (e.g., adding a signature), iii) the constant opportunity cost with a minimum constant reward. In practice, rational validators account for these costs beforehand and conclude that participation is beneficial. (b)~A rational validator considers a strategy only if the resulting utility is guaranteed within the ledger, as we do not account for external trustless mechanisms. 
Promises or unverifiable commitments from the adversary regarding future actions are not considered credible and, thus, will not influence the validator's strategy.
Given the appropriate incentives, however, rational validators may collude with other rational and Byzantine validators. (c)~A rational validator deviates from the protocol only if such deviation results in strictly greater utility, as devising a strategy to deviate may incur additional costs.
For example, a rational validator will censor or double spend only if they can secure a bribe larger than their potential loss due to slashing.

\vspace{-0.15cm}
\subsection{Definitions and Properties}

\com{
\chris{\begin{itemize}
    \item I tried to describe the properties in words, as it seems they like it 
    \item do we have to clarify anything regarding the "PoS" based blockchain protocols?
    \item maybe citations should be added \end{itemize}}

\chris{Rational behind the definitions: As I understand it, the goal of an SMR is to provide the guarantees to clients. Correct validators should agree, so that it is possible for clients to extract valid responses. Following that I didn't think of giving the same guarantees we have for correct validators to rational validators, i.e., they output the same transaction ledger. They can do whatever they want if they choose to deviate from the protocol. I had in mind the following. Total order says: correct validators will still agree, misbehaving validators cannot manage to create disagreements. And then, the client-security, guarantees that clients only extract confirmations from correct validators (almost). So, we are still fine, while keeping what SMR fundamentally is.} }

At a high level, in an SMR protocol, validators must perform two primary tasks: order a set of transactions which they execute to extract the state, and reply to clients. 
Traditional SMR protocols such as PBFT~\cite{pbft} and HotStuff~\cite{yin2018hotstuff} assume an honest supermajority of validators -- specifically $2f+1$ out of the total $3f+1$ validators must be honest. Under this assumption, clients can extract the confirmation of transactions simply by waiting for responses from $f+1$ validators. In contrast, our setting introduces rational validators and assumes only an \emph{honest minority}. For this reason, we separate these two responsibilities: (i) ordering transactions, and (ii) replying to clients.

\vspace{3pt}
\noindent\textbf{SMR problem.}
Given as input transactions submitted by clients, the validators execute an SMR protocol $\Pi$ and output an ordered sequence of confirmed transactions, known as the transaction ledger. We say that $\Pi$ solves the SMR problem when: (i) correct validators agree in a \emph{total order}~\cite{defago2004total}, (ii) the transaction ledger is \emph{certifiable} to clients.

\begin{definition}
    We say that correct validators agree in a \textbf{total order}, when the transaction ledger output by $\Pi$, for every execution, satisfies: 
    \end{definition}
\begin{itemize}[nosep,leftmargin=*]
    \item \textbf{Safety}. The transaction ledgers of any two correct validators are prefixes of one another,
    \item \textbf{Liveness}. Any valid transaction witnessed by all correct validators, will eventually be included in the transaction ledger of every correct validator. 
\end{itemize}

\begin{definition}\label{def:cert}
    We say that $\Pi$ satisfies \textbf{certifiability} if the following holds for every execution, 
    \begin{itemize}[nosep,leftmargin=*]
        \item \textbf{Client-safety}: Clients accept the confirmation of a transaction only if it is included to the transaction ledger of at least one correct validator,
        \item \textbf{Client-liveness}: Clients accept confirmation for any transaction included to the transaction ledger of all correct validators. 
    \end{itemize}
\end{definition} 

We stress that \emph{total order} and \emph{certifiability} are distinct guarantees. Some protocols achieve total order but not certifiability (e.g., Dolev--Strong~\cite{dolev1983authenticated} for $f > n/2$ Byzantine faults). 
To satisfy both properties, validators and clients may adopt different confirmation rules, in line with the flexible BFT paradigm~\cite{malkhi2019flexible}. 

\vspace{3pt}
\noindent\textbf{Resiliency of SMR protocols.} Traditional SMR protocols~\cite{hotstuff, pbft, gelashvili2022jolteon} are secure when considering $f$ \emph{Byzantine} validators and $n-f$ \emph{correct} validators. We say that a protocol satisfying this condition is \emph{$(n,f)$-resilient}. 
 
\begin{definition}
A protocol \( \Pi \) executed by \( n \) validators is \textbf{\((n,f)\)-resilient} if it remains secure in the presence of up to \( f \) Byzantine validators, assuming that the remaining  $n - f$ validators follow the protocol correctly. 
\end{definition}

However, in traditional $(n,f)$-resilient SMR protocols, one rational validator is enough to violate safety and liveness as Byzantine validators can bribe the rational validator to collude with them. 
To capture both byzantine and rational behavior in our protocol, we introduce the notion of $(n, k, f)$-resiliency.
\begin{definition}
A protocol \( \Pi \) executed by \( n \) validators is \textbf{\((n, k, f)\)-resilient} if it remains secure in the presence of \( f \) Byzantine and \( k \) rational validators, assuming that the remaining $n - f - k $ validators follow the protocol correctly. 
\end{definition}

\vspace{3pt}
\noindent\textbf{$q$-commitable SMR protocols.} 
We focus on a family of SMR protocols that progress by collecting certificates of $q$ votes across multiple phases, similar to several BFT SMR protocols (e.g.~\cite{hotstuff, buchman2018latest}). 
 We define a \emph{quorum certificate} (QC) as a certificate consisting of \( q \) votes. A protocol is \emph{q-commitable} if a QC on a specific message \( m \) (e.g., \( m = \texttt{commit} \) in PBFT~\cite{pbft}) is both a necessary and sufficient condition for validators to commit a transaction (or block) $b$ to their local ledger. 

 


\begin{definition} A protocol $\Pi$ is a \textbf{$q$-commitable} SMR protocol when every correct validator includes a block $b$ in its local transaction ledger if and only if it has witnessed a QC on a specific message $m$, denoted $QC_{m}(b)$.
\end{definition} 

For simplicity, we may often omit the specific message $m$ and refer to the certificate that finalizes $b$ as $QC(b)$.
We stress that \( QC(b) \) may not be sufficient for certifiability from the clients' perspective, as they do not actively monitor the SMR execution (i.e., they are silent). However, for clients that monitor the SMR execution (e.g., the validators), \( QC(b) \) is sufficient for certifiability.


We study SMR protocols in the context of proof-of-stake (PoS) blockchains, where validators batch transactions into blocks, as in HotStuff~\cite{yin2018hotstuff}. Our objective is to determine when an \((n,f)\)-resilient SMR protocol can be transformed into an \((n,k,f)\)-resilient PoS blockchain.  
Specifically, we present a transformation that converts any \((n,f)\)-resilient \( q \)-certifiable PoS blockchain~\cite{pbft,hotstuff,gelashvili2022jolteon} into an \((n,k,f)\)-resilient PoS blockchain. 
At the core of this transformation, we introduce a new certification mechanism for finalizing transactions, ensuring security under our hybrid model.

\vspace{3pt}
\noindent\textbf{Finality latency.} 
\emph{Finality latency} refers to the time that elapses between a transaction being submitted by a client and its irreversible commitment to the transaction ledger. 

\vspace{3pt}
\noindent\textbf{Recovery and economic restitution.} 
A protocol that is $(n,k,f)$-resilient constitutes a complete solution to SMR: when at most $f$ validators are Byzantine and at most $k$ are rational, the protocol preserves both safety and liveness.  
We now consider executions in which the actual number of Byzantine validators $f^*$ exceeds the resilience threshold $f$ of an $(n,f)$-resilient system. In such executions, the protocol may suffer both safety and liveness violations.

Liveness attacks cannot, in general, be deterred, since Byzantine validators may simply refuse to participate. For this reason, the literature on recoverable consensus typically assumes an adversary that is willing to participate in the protocol (e.g., to earn participation rewards) but may opportunistically violate safety; this type of faults are called \emph{a-b-c} Byzantine~\cite{malkhi2019flexible,zlb, gong2025recover}. In the same spirit, we focus only on safety attacks with economic impact, namely conflicts among blocks that are accepted by clients\footnote{Recall that the client finalization rule might differ to the validator rule.}. Our goal is to ensure that, even if a client-side safety violation occurs, the system can recover and resume execution from a consistent state in which all honest clients harmed by the violation are \emph{economically reimbursed}, without minting new coins\footnote{Minting new coins inflates the currency thus penalizing honest clients.}. 

\begin{definition}\label{def:economicDef}
We say that an SMR protocol \emph{regains \textbf{economic restitution} after a safety violation} with recovery parameter $\Delta_R$ if, for every execution, after any safety violation, the protocol satisfies: 
\ifsubmissionversion
\emph{(i) Deterministic recovery}: all honest validators eventually recover to a common state $S$ for which there exists a verifiable state certificate that every honest validator witnesses at most $\Delta_R$ time after the equivocation,
\emph{(ii) Accountability:} in state $S$, validators responsible for the equivocation are punished, losing their stake, 
\emph{(iii) Client reimbursement:} in state $S$, each honest client receives the full monetary value of every transaction it has accepted during the execution,
\emph{(iv) Monetary conservation:} the total supply of coins remains the same (new coins are not minted).

\else
\fi
\end{definition}

\subsection{SMR Game and Coalition-Compliance} 
To connect game-theoretic rational behavior to protocol-level resilience guarantees, we introduce \emph{$(n, k, f)$-coalition-compliance}. Informally, a protocol is \emph{$(n, k, f)$-coalition-compliant} if no rational validator prefers to follow a strategy that violates the SMR protocol, regardless of the actions of Byzantine validators. This property holds even if all rational and Byzantine validators collude. 

By transforming an $(n,f)$-resilient protocol into an $(n, k, f)$-coalition-compliant protocol, we ensure that rational validators have no incentive to violate the SMR security properties. As a result, the transformed protocol effectively achieves $(n,k,f)$-resilience, even when rational and Byzantine parties collude. 

\begin{definition}[SMR Game]
The \textbf{SMR game} is a tuple:
\[
G = (n, (S_i)_{i \in [n]}, (u_i)_{i \in [n]}, H), \quad\text{where:}
\]
\begin{itemize}
    \item \( n \in \mathbb{N} \) is the number of players. Each player is either \emph{correct}, \emph{Byzantine}, or \emph{rational}. 
    
    \item \( H \) is the global history of the protocol execution. Each player \( i \in [n] \) observes a local view \( H_i \subseteq H \).
    
    \item \( S_i \) is the set of all possible strategies available to player \( i \), where a strategy is a function mapping local histories \( H_i \) to actions.
    A global strategy profile is \( s = (s_1,\dots, s_n) \in S = \prod_{i \in [n]} S_i \).
    
    \item The utility function \( u_i : S \to \mathbb{R} \) gives the total payoff to player \( i \) over the entire execution.

\end{itemize}
We denote by \( s_{-i} \) the strategy profile excluding player \( i \), and by \( S_{-i} \) the set of all such profiles.
\end{definition}

We partition the space of strategy profiles into two sets: \(S_{\text{bad}}\), consisting of strategy profiles resulting in safety violations (i.e., conflicting blocks are finalized with valid certificates) or liveness violations (i.e., permanent censorship of a transaction), and $S_{\text{good}}$ which preserve both properties.

For each validator \( p \), we define \( S_{\mathsf{bad}}^p \) as the set of individual strategies \( s_p \) for which there exists a strategy profile \( s = (s_p, s_{-p}) \in S_{\mathsf{bad}} \) where \( p \) actively contributes to the violation. Conversely, 
for each validator \( p \), we denote by \( S_{\mathsf{good}}^p \) the set of individual strategies \( s_p \) such that there exists a strategy profile \( s = (s_p, s_{-p}) \in S_{\mathsf{good}} \).


We say that a protocol is \emph{coalition-compliant} if no rational validator can profit
from contributing to safety or liveness violations, even under full collusion with
Byzantine validators. In other words, rational validators always have a protocol-compliant strategy that is equally good or better than any deviation causing violations.


%

\begin{definition}  \label{weaklyDominate}
A strategy \( s_i \in S_i \) \textbf{weakly dominates} a strategy \( s_i' \in S_i \) if, for every possible strategy profile \( s_{-i} \in S_{-i} \), it holds that 
$u_i(s_i, s_{-i}) \geq u_i(s'_i, s_{-i}),$
and there exists at least one \( s_{-i} \in S_{-i} \) such that:
$u_i(s_i, s_{-i}) > u_i(s'_i, s_{-i}).$
\end{definition}

\begin{definition}
We say that an SMR game is \textbf{$(n, k, f)$-coalition-compliant} if for any rational validator \( p \), and every strategy \( s_p \in S_{\mathsf{bad}}^p \), there exists a strategy \( s_p' \in S_{\mathsf{good}}^p \) that weakly dominates \( s_p \).
\end{definition}

Coalition-compliance strengthens the compliance notion of~\cite{karakostas2022blockchain} by permitting collusion between Byzantine and rational validators, and is therefore stronger than an equilibrium notion in which rational parties act independently. It is deliberately weaker than incentive compatibility, which would require protocol-following to be utility-maximizing against every adversarial profile. That requirement is unattainable here: a rational validator that withholds a transaction in rounds where the offered bribe exceeds the associated fee earns strictly more than one that includes it immediately, so protocol-following is not utility-maximizing. Coalition-compliance instead guarantees that such deviations only delay inclusion; every valid transaction is eventually finalized, so liveness is preserved, while no safety violation or permanent censorship is ever strictly profitable. Coalition-compliance thus targets the minimal conditions for safety and liveness rather than pinning down a unique strategy.

\section{Impossibility results} \label{sec:impossibilitiesSummary}
\za{Can I cut some of the 3.1-3.7??}
In the following, we identify settings in which a $q$-committable \((n,f)\)--resilient PoS blockchain cannot be transformed into an \((n,k, f^*)\)--resilient PoS blockchain for $f^*\leq f$. 
Additionally, we determine when $q$-committable SMR protocols fail to guarantee finality before the synchrony bound $\Delta^*$ elapses. 

We provide a summary of our results and defer the proof sketches to Appendix~\ref{impossibilities:sketch} and the formal proofs to  
\ifsubmissionversion
Appendix~G extended version~\cite{fullversion}.
\else
Appendix~\ref{sec:impossibilities}.
\fi We also illustrate the main results in Figure~\ref{fig:designspace}.  We stress that our impossibility results apply to the setting of \emph{silent clients}, where clients decide whether to accept a block based solely on an accompanying certificate, without engaging in additional communication rounds with other participants. Finally, $(n,k,f)$--coalition compliance is a deterministic guarantee that must hold in every execution; thus, a single attack suffices to prove impossibility. We discuss probabilistic guarantees in Section~\ref{sec:discussion}.

\begin{figure*}[t]
\centering

\begin{minipage}[t]{0.48\textwidth}
\centering
\includegraphics[width=\linewidth]{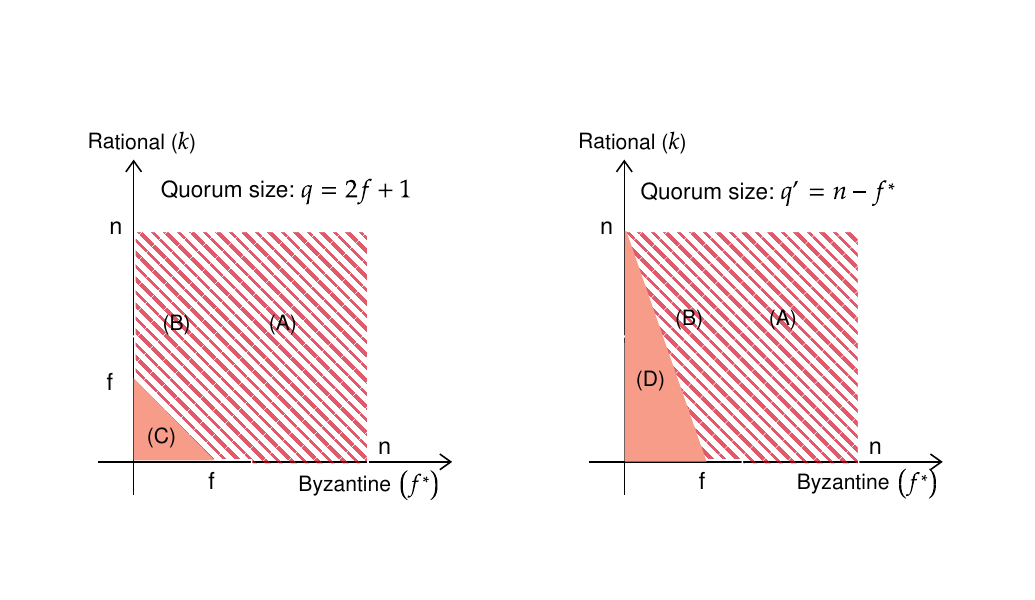}

\footnotesize\textbf{Partial synchrony}

\scriptsize
\raggedright
(A): Impossibility for \(f^* > \lceil h/2 \rceil\) (Theorem~\ref{impossibility:synchrony}), \\
(B): Impossibility for  $q \leq f + k +  \lfloor \frac{h}{2} \rfloor$(Theorem~\ref{impossibility:partialSync2}), \\ (C): \((n,f)\)-resilient SMR protocols are also \((n,f^*,k)\)-resilient for  \(f^* + k \leq f\), \\
(D): Enlarging the quorum to \(q = n - f^*\) allows resilience for all \(k + 3f^* < n\).
\end{minipage}
\hfill
\begin{minipage}[t]{0.48\textwidth}
\centering
\includegraphics[width=\linewidth]{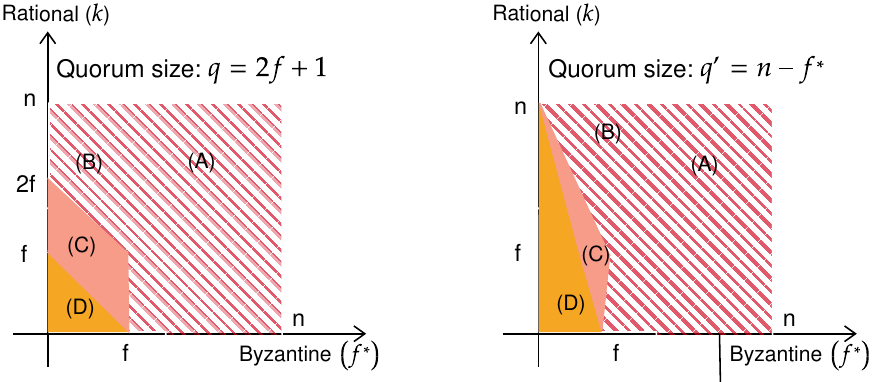}

\footnotesize\textbf{Synchrony}

\scriptsize
\raggedright
(A) We focus on $q$--commitable protocols that still remain safe in partial synchrony with up to $f^*$ Byzantine validators, \chris{fix}\\
(B) Impossibility for \(f^* + k \geq q\) (Theorem~\ref{impossibility:correctvalidators}), \chris{fix $[n/3, n/2]$}\\
(C), (D): Similar to the left side. \\ 
(E) Possibility (\sys) only by limiting transaction volume (Theorem~\ref{impossibility:responsiveness}). 
\end{minipage}

\caption{Design space of \(2f+1\)-committable SMR protocols.}

\label{fig:designspace}

\end{figure*}

\vspace{3pt}
\noindent\textbf{Partial synchrony.}
In Theorem~\ref{impossibility:synchrony}, we show that designing an  $(n, k, f)$-resilient $q$-commitable SMR protocol is impossible for $ f \geq \lceil \frac{h}{2} \rceil$ ($h$ is the number of correct validators). Notice that this result depends only on the proportion of Byzantine and correct validators. This is because to guarantee liveness, the size of the quorum $q$ is at most $n-f$. For $ f \geq \lceil \frac{h}{2} \rceil$, we cannot guarantee that two quorums intersect in at least one correct validator. Rational validators, in turn, will collude with Byzantine validators to form quorums of conflicting blocks to maximize their monetary gains.


To establish this result, we extend the proof of Budish et al.~\cite[Theorem 4.2]{budish2024economic} within our model. The core issue is that rational validators can collude with Byzantine to generate valid quorum certificates while avoiding penalties. This is possible because, for any predefined withdrawal period that allows validators to reclaim their stake, correct validators cannot reliably distinguish between an honest execution with network-induced delays and malicious behavior.
Hence, rational validators may engage in equivocation without facing slashing, undermining the protocol's security. 

\begin{restatable}{theorem}{PartialSynchrony} \label{impossibility:synchrony}
Assuming silent clients, a static setting, and a partially synchronous network, no $q$-commitable SMR protocol can be \((n, k, f)\)--resilient for any \( q \) when \( f \geq \lceil \frac{h}{2} \rceil \), where \( h = n - k - f \).

\end{restatable}

It follows directly from Theorem~\ref{impossibility:synchrony} that no \( q \)-commitable protocol can be \( (n,f) \)--resilient when \( f = \lceil h/2 \rceil \), which agrees with the well-known result that resilience is impossible once Byzantine validators exceed one third of the network (\( f \geq n/3 \)). Furthermore, our result shows that even if a \( q \)-commitable protocol is resilient for \( f = \lfloor h/2 \rfloor \), replacing just one correct validator with a rational one breaks resilience if the total number of Byzantine and rational validators reaches at least one third the network (\( f + k \geq n/3 \)), as shown in Corollary~\ref{impossibility:synchrony2}. This result highlights the fragility of quorum-based protocols under rational behavior, as even a small shift from correct to rational participants significantly impacts security.\za{to shorten}

\begin{corollary} \label{impossibility:synchrony2}
Consider any $q$-commitable protocol $\Pi$ that is $(n, f)$--resilient for $f = \lfloor \frac{h}{2} \rfloor$ where $h=n-f$ is the number of correct validators. For any $k\geq 1$, $\Pi$ is not $(n, k, f)$--resilient in our model under partial synchrony. 
\end{corollary}


Corollary~\ref{impossibility:synchrony2} implies that no quorum-based SMR protocol \( \Pi \) can achieve \((3f+1, k, f)\)--resilience for $k \geq 1$ in the partial synchronous model. That means that quorum-based SMR protocols are not resilient for $f+k \geq n/3$ under the model described in Theorem~\ref{impossibility:synchrony}.

\vspace{3pt}
\noindent\emph{Fewer Byzantine.} Theorem~\ref{impossibility:synchrony} characterizes the maximum Byzantine resilience of \((n,k,f)\)-resilient SMR protocols, independently of the quorum size \(q\). A natural question is what happens \emph{below} this maximum Byzantine threshold. Can an SMR protocol remain secure when \(f^* < \lceil h/2 \rceil\) and the combined number of Byzantine and rational validators exceeds one third of the network, i.e., \(f^* + k > n/3\)?

When \(f = \lceil h/2 \rceil\), even a single rational validator can collude with Byzantine validators to equivocate without being detected by other rational validators. However, for \(f^* < \lceil h/2 \rceil\), rational validators must collude to equivocate. This opens a potential design space in which rational validators can be incentivized to expose misbehavior. One such approach is a \emph{baiting strategy},  leveraged by TRAP~\cite{trap} to achieve consensus on a single value, where rational validators are rewarded for reporting equivocation. Theorem~\ref{impossibility:partialSync2} characterizes the conditions under which, even in the presence of baiting strategies, solving SMR remains impossible.

\begin{restatable}{theorem}{partialSync2} 
\label{impossibility:partialSync2}
Assuming participation rewards are distributed to validators, silent clients, a static PoS setting, and a partially synchronous network. Without any bounds i) on the liquid coins held by rational validators or malicious clients, ii) or on the participation rewards, the following holds: no $q$-commitable SMR protocol can be \((n, k, f)\)--resilient for any \( q \) when \( q \leq f + k +  \lfloor \frac{h}{2} \rfloor \), where \( h = n - k - f \). \end{restatable}

Theorem~\ref{impossibility:partialSync2} implies that no quorum-based SMR protocol can achieve 
$(3f+1, k, f^*)$--resiliency with $k = f - f^*+1$ in partial synchrony without enlarging the quorum size.

\begin{corollary} \label{impossibility:synchrony3}
Consider any $q$-commitable protocol $\Pi$ that is $(n, f)$--resilient for $f = \lfloor \tfrac{h}{2} \rfloor$, where $h = n-f$ is the number of correct validators. For any $ f^* < f$ and $k = f - f^*+1$, the protocol $\Pi$ is not $(n, k, f^*)$--resilient in our model under partial synchrony, without enlarging the quorum size.
\end{corollary}

Thus, $2f+1$--commitable quorum-based SMR protocols with a quorum of size $q=2f+1$ are not resilient for any 
$f^* + k \geq n/3$ under the assumptions of Theorem~\ref{impossibility:partialSync2}. 

\chris{possibly have some proofs in the appendices, because this is not obvious}
\noindent\emph{Closing the gaps.}  Corollary~\ref{impossibility:synchrony3} establishes that, in partial synchrony, reducing the actual number of Byzantine validators does not improve resilience unless we increase the quorum size. Following the flexible BFT approach~\cite{malkhi2019flexible}, in this setting, we can enlarge the quorum size used by clients from the classical $q = n-f$, to a higher quorum $q' = n-f^*$, where $f^*$ is the actual number of Byzantine validators present. Then, the constraint for rational validators is $k < n - 3f^* $; this choice ensures that the number of honest parties $ h \geq n - 2f^*$, ensuring that $q' = n +f^* > f^* + k + \lfloor h/2 \rfloor $, so that any two quorums intersect in at least one honest validator and equivocation becomes impossible. Liveness is guaranteed by the incentives of rational validators, who are motivated to follow the protocol in order to obtain participation rewards.



\vspace{5pt}
\noindent\textbf{Synchrony.} To circumvent the impossibility result of Theorem~\ref{impossibility:synchrony}, we assume synchrony with a known time delay $\Delta^*$. However, we now show that even in synchronous settings, designing resilient SMR protocols remains impossible for $f + k \geq h$. Here, a client cannot distinguish whether correct validators participated in a quorum certificate. Consequently, Byzantine and rational validators can generate certificates that the clients will accept (to ensure client-liveness) without any participation from correct validators.

\begin{restatable}{theorem}{Synchrony} \label{impossibility:correctvalidators}
Assuming silent clients, a static setting, and a synchronous network, no $q$-commitable SMR protocol can be $(n, k, f)$-resilient, when $f + k \geq h$, where $h= n-k-f$.
\end{restatable}



By setting \( k = 0 \) in Theorem~\ref{impossibility:correctvalidators}, we conclude that no $q$-commitable protocol can be \((n,f)\)--resilient when \( f = h \). Furthermore, Theorem~\ref{impossibility:correctvalidators} establishes that even if a $q$-commitable protocol remains \((n,f)\)--resilient for \( f = h - 1 \), we cannot replace a single correct validator to a rational one, implying that resilience breaks down when \( f + k \geq 2n/3 \). This result is formalized in Corollary~\ref{impossibility:correctvalidators2}.

\begin{corollary} \label{impossibility:correctvalidators2} 
Consider any $q$-commitable protocol $\Pi$ that is $(n, f)$--resilient for $f = h-1$ where $h=n-f$ is the number of correct validators. For any $k\geq 1$, $\Pi$ is not $(n, k, f)$--resilient in our model under synchrony. 
\end{corollary}

In other words, no quorum-based SMR protocol can achieve \((n, f+1, f)\)--resilience, even in the synchronous model. Therefore, quorum-based SMR protocols are not resilient for $f+k \geq 2n/3$. 

\noindent\emph{Beyond $n/3$ Byzantine.}
Synchronous SMR protocols tolerate up to $f^* \in [n/3,n/2)$ Byzantine validators~\cite{abraham2020sync}. Since their safety relies on synchrony, while we focus on preserving $(n,f)$-resilience under partial synchrony and use the synchrony bound only to strengthen fault tolerance when the number of failures exceeds $f$, we omit them from Figure~\ref{fig:designspace}(also see Related Work).
\za{TODO: Shorten a bit!}

\vspace{-0.15cm}
\subsubsection*{\textbf{Finality latency}} 

Theorem~\ref{impossibility:responsiveness} establishes that without bounds on the liquid coins or the participation rewards, no protocol can ensure finality during a $\Delta^*$-bounded window when $f \geq \lceil \frac{h}{2} \rceil$  or \( q \leq f + k +  \lceil \frac{h}{2} \rceil \). That is because Byzantine and rational validators can still produce valid quorum certificates for conflicting blocks within $\Delta^*$, with each rational validator double-spending an amount exceeding their stake. This attack remains profitable for rational validators, even if they get slashed and forfeit their stake.

\begin{restatable}{theorem}{responsiveness} 
\label{impossibility:responsiveness}
Assuming participation rewards are distributed to validators, silent clients, a static PoS setting, and a synchronous network. Without any bounds i) on the liquid funds held by rational validators or malicious clients, or ii) on the participation rewards, the following holds:
No SMR protocol can be both (i) $(n, f)$--resilient $q$-commitable with finality latency $t$, and (ii) $(n, k, f)$--resilient with finality latency $t'<t+\Delta^*$, when $f \geq \lceil \frac{h}{2} \rceil$ or \( q \leq f + k +  \lfloor \frac{h}{2} \rfloor \), for $h= n-k-f$, for any execution.
\end{restatable}

According to Theorem~\ref{impossibility:responsiveness}, if we run an  $(n,f)$--resilient quorum-based SMR protocol \( \Pi \) with $ f \geq \lceil \frac{h}{2} \rceil $ validators or \( q \leq f + k +  \lfloor \frac{h}{2} \rfloor \), it is impossible to finalize every transaction before  $\Delta^*$ elapses. That means that we must limit the transaction volume finalized per any $\Delta^*$ time window.
\section{\sys: Finalization Rule }
\label{sec:protocol}

In this section, we present the \sys finalization protocol, a transformation that converts any $(n,f)$-resilient $(2f+1)$-committable SMR protocol~$\Pi$ into an $(n,k,f^*)$-resilient proof-of-stake SMR protocol for $n = 3f+1$.  

Throughout this section, we assume a synchronous network with message delay bounded by~$\Delta^*$, as required by Theorem~\ref{impossibility:synchrony}. We consider executions in which the number of Byzantine and rational validators satisfy $f^* \leq f$ and $k\leq 2f - f^*$, respectively. Our protocol achieves the optimal resilience bounds established in Theorem~\ref{impossibility:correctvalidators}.
Recovery from executions exceeding this Byzantine threshold is deferred to Section~\ref{sec:recovery}.

\vspace{3pt}
\noindent\textbf{Overview.}
As established in Section~3, in the presence of rational validators and silent clients, quorum-based SMR protocols are vulnerable to two major attack vectors:
(i) \emph{safety attacks}, where rational validators collude with Byzantine ones to produce conflicting blocks with valid quorum certificates within the window~$\Delta^*$, enabling double-spending attacks; and
(ii) \emph{censorship attacks}, where the adversary uses accumulated participation rewards to bribe rational validators into withholding votes for blocks containing targeted transactions, preventing their finalization.

To address these challenges, our protocol introduces the following components:

\begin{itemize}[nosep,leftmargin=*]
    \item \textbf{A finality gadget} that ensures \emph{safety} by enforcing {stake-bounded finalization}: within any $\Delta^*$-bounded window, the maximum transaction volume that can be finalized is limited by the total stake at risk; making equivocations unprofitable once \textbf{slashing penalties} are applied.
    This ensures that rational validators have no incentive to fork. 
    \item \textbf{A reward scheme} that incentivizes rational validators and mitigates censorship (\emph{liveness}): fees are distributed only in execution windows where at least $n-f$ distinct validators contribute to block production.
    \item A \textbf{strong chain mechanism} that enables higher-volume finalization during periods of high participation, improving \emph{efficiency}.
\end{itemize}

\subsection{Finality Gadget}
The \emph{finality gadget}  ensures that validators finalize only transactions that cannot be reversed and enables them to provide inclusion proofs to clients. Below, we describe the validator and client protocols for accepting transactions. Since clients do not monitor the protocol execution, their confirmation rule differs from that of validators. As a result, validators may finalize transactions based on protocol guarantees, whereas clients rely on proofs to verify finality.


\vspace{3pt}
\noindent\textbf{Validator confirmation.}
We assume a  $q$-committable SMR protocol $\Pi$ that is broadcast-based\footnote{Otherwise, validators in our construction can broadcast the valid certificates to propagate the required information.} 
to ensure that safety violations are detected by every correct validator within \(\Delta^*\). We treat $\Pi$ as a black-box abstraction:
%
whenever a validator obtains a valid quorum certificate $QC(b)$ for a block~$b$ at height~$h$, the protocol outputs $\mathsf{deliver}(QC(b), b, h)$ to that validator (line~\ref{alg:line:consensusoutput}). Then, the validator 
%
stores the certificate (line~\ref{alg:line:quorum}) along with the block and its local timestamp (line~\ref{alg:line:ledger}), and initiates the block finalization function. 

Since conflicting blocks may exist but remain unwitnessed within the last $\Delta^*$, validators employ the \emph{stake-bounded finalization rule}. More specifically, each validator $p$ traverses its local ledger and using the timestamp of the blocks and its local clock, determines which blocks are \( \Delta^* \)--old  (output $\Delta^*$ ago or more), or \( \Delta^* \)--recent (output within the last $\Delta^*$)\footnote{To measure elapsed time, the validator relies on its local clock rather than requiring synchronized clocks (lines \ref{alg:line:ledger}, \ref{alg:line:iterateledger}).}.
Then, $p$ finalizes blocks as follows.

\noindent\emph{\( \Delta^* \)--old Blocks.}  If no conflicting block has been detected, validators finalize any block \( b' \) that was output by \( \Pi \) at least \( \Delta^* \) ago (lines \ref{alg:line:deltafinal} and \ref{alg:line:finalizeD}). 

\noindent\emph{\( \Delta^* \)--recent Blocks.}  
For blocks output within the last \( \Delta^* \), each validator finalizes only a subset of the most recent blocks, ensuring that the total finalized transaction volume does not exceed \( C = D \) coins (lines \ref{alg:line:collateralfinal} and \ref{alg:line:finalizeC}). This guarantees that any rational validator engaging in an attack can gain at most \( D \) coins. 
 For intuition, we provide a tight example in Figure~\ref{fig:forkingAttack} for $f^*=f$ Byzantine and $k=f$ rational validators. 
\begin{figure}[h!]
    \centering
    \includegraphics[width=0.8\columnwidth]{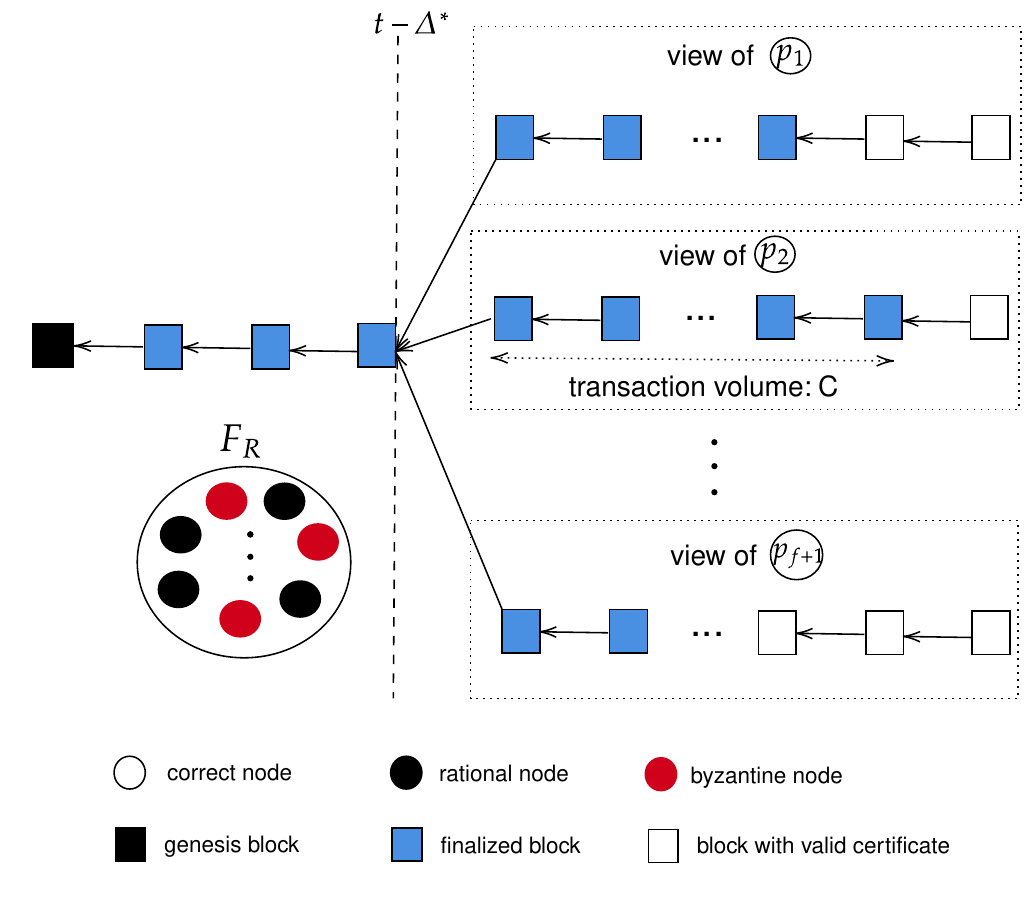}
    \caption{Forking attack.}
    \label{fig:forkingAttack}
    
    \raggedright
    \footnotesize A coalition of $2f$ misbehaving validators (Byzantine and rational) fork the system. The adversary partitions all correct validators into $f+1$ distinct sets and collects their signatures to create conflicting blocks with valid certificates before $\Delta^*$ elapses. Each correct validator finalizes blocks with transaction volume up to $C = D$ within  $\Delta^*$. 
    Assuming Byzantine validators claim no rewards and instead use them to bribe the \( f \) rational validators they need to successfully double-spend, rational validators share a total of \( f \times D \) double-spent coins, resulting in a per-validator gain of \( D \). 
    
\end{figure}

\vspace{3pt}
\noindent\textbf{Finality votes.}
To enable inclusion proofs and ensure accountability, validators include \emph{finality votes} for a block \(b\) in subsequent blocks once \(b\) is finalized. We say that \(b\) has a \emph{finality certificate} \(\mathcal{F}(b)\), if at least \(2f+1\) finality votes for \(b\) (or for any block extending \(b\)) are recorded on-chain, with each vote included in a block carrying a valid $QC$.

\vspace{3pt}
\noindent\textbf{Conflicting certificates.} Correct validators never commit finality votes for conflicting blocks. If a validator observes conflicting blocks with valid certificates or finality certificates (lines \ref{alg:line:violation1}, \ref{alg:line:violation2}), it immediately stops finalizing new blocks (lines \ref{alg:line:deltafinal}, \ref{alg:line:collateralfinal}). By limiting the transaction volume so that the total potential gain for any rational validator does not exceed their stake, we ensure that, due to slashing, rational validators have no incentive to engage in equivocations.

\vspace{3pt}
\noindent\textbf{Client confirmation.} Clients, who do not participate directly in the consensus protocol, cannot observe when validators' internal confirmation rules are met. The finality certificate thus serves as a verifiable proof of finality, allowing clients to confirm that a block has been finalized according to the validators' rules. Specifically, to confirm the inclusion of a transaction \( tx \), validators provide clients with a proof consisting of an inclusion proof (e.g., a Merkle proof) for \( tx \) in block \( b \), and a finality certificate \( \mathcal{F}(b) \) for \( b \), including the finality votes and the valid quorum certificates of the blocks containing these votes.



In this way, clients can securely verify transaction finalization, ensuring that the SMR protocol is \emph{certifiable}.

\subsection{Rewards and Penalties}\label{sec:slash}
\noindent\textbf{Withdrawal.} 
To ensure misbehaving validators are penalized, we implement a two-phase withdrawal process~\cite{ethereumwithdrawal}:
\begin{itemize}[nosep,leftmargin=*]
    \item Validators must stake \( D \) coins to participate in \( \Pi \).
    \item Withdrawals are only allowed after a period \( \Delta_W > 2\Delta^* \), during which validators cannot participate in the protocol. 
    \item If a validator is detected signing finality votes for conflicting blocks
    , it is blacklisted (line \ref{alg:line:blacklist}), preventing its withdrawal permanently. 
\end{itemize}

Note that the withdrawal period is set to be larger than $2\Delta^*$ (instead of $\Delta^*$). Although equivocations are observed within $\Delta^*$, which is sufficient to limit the transaction volume, only after $2\Delta^*$ can we be certain that no conflicting block exists. E.g., consider a block finalized at time $t$ by $2f$ rational and Byzantine validators along with one honest validator $p_i$
. Just before time $t + \Delta^*$, the $2f$ malicious validators may finalize a conflicting block with the help of another honest validator $p_j$. In the worst case, $p_i$ will become aware of the conflict at time $t + 2\Delta^*$. Setting the withdrawal period larger than $2\Delta^*$ is sufficient to prevent malicious validators from withdrawing their funds prematurely.

\vspace{3pt}
\noindent\textbf{Rewards.}
Participation rewards are sourced only from transaction fees. Our reward scheme distributes them to validators only in windows where enough distinct validators participate, and scales each payout by the total participation gathered, so that rational validators are incentivized to propose, vote, and include the votes of others rather than censor them. 

In particular, rewards are distributed every $w$ rounds, where $w$ is chosen such that every validator is scheduled as leader at least once within each window. We define the interval of rounds $[kw+1,(k+1)w]$ for $k\in\mathbb{Z}$ as the leader window $W_k$. Let $L_{W_k}$ denote the number of unique validators that successfully propose blocks with valid certificates during $W_k$. To ensure that at least one correct validator successfully proposes a block, rewards are distributed only if $L_{W_k}\geq n-f$, i.e., at least $n-f$ unique validators successfully submit proposals during the leader window.

We define $P_b$ as the set of validators participating in the valid SMR certificate of some block $b$, $P_{W_k}$ as the set of unique validators participating in at least one finalized block during $W_k$, and $\mathcal{B}_k$ as the set of all blocks finalized in $W_k$. The aggregated participation factor of window $W_k$ is $\alpha_k=\min\!\left\{\sum_{b\in\mathcal{B}_k}|P_b|+\epsilon,\; n\right\}$, where $\epsilon>0$ is a protocol parameter. Let $F_{W_k}$ denote the total transaction fees collected from finalized transactions during $W_k$. The participation reward pool for window $W_k$ is $R_{W_k}=F_{W_k}\cdot\frac{\alpha_k}{n}$, and each validator in $P_{W_k}$ receives an equal share $\frac{R_{W_k}}{|P_{W_k}|}$.

This mechanism guarantees censorship resilience, as Byzantine validators can only earn rewards in windows that include at least one honest proposer. Even if a Byzantine validator participates normally to accumulate rewards and bribe rational validators later, it eventually depletes its coins. With no coins left to bribe, rational validators prefer to follow the protocol and collect the participation rewards. Furthermore, the participation factor $\alpha_k$ incentivizes rational validators to include as many valid votes and certificates as possible. Looking ahead, validators are incentivized to build strong chains (Section~\ref{sec:strongRules}) and to avoid censoring votes from other rational validators, since excluding their participation lowers $\alpha_k$ and thus the total rewards distributed within the leader window.

We discuss in Section~\ref{sec:discussion} how this rewards scheme generalizes to q-commitable is that either the leader is elected uniformly at random~\cite{buchman2018latest} or they are multi-proposer~\cite{danezis2022narwahl}.



\subsection{Extending Finalization with Strong Chains}\label{sec:strongRules}
Recall that validators finalize transaction volumes up to \( D \) coins within a time window \( \Delta^* \). 
This is the lower bound in our threat model. When the system is not under attack (i.e., fewer faulty nodes), blocks often receive more than \( 2f+1 \) signatures. We leverage the certificates with additional signatures to further limit forks and finalize larger transaction volumes without waiting for \( \Delta^* \).  

To achieve this, we extend the concept of \emph{\( i \)-strong chains}~\cite{xiang2021strengthened} in our model. A ledger is \textbf{\( i \)-strong} if it contains valid certificates signed by \( 2f+1+i \) unique validators within \( \Delta^* \). With at least \( 1+i \) correct validators, the maximum number of forks decreases to \( f-i \), allowing secure finalization of higher transaction volumes.
In particular, when a certificate has \( 2.5f+1 \) signatures, we can define a \emph{strongest chain} that can finalize \emph{infinite} transaction volume (regardless of $D$).





\begin{definition}[$i$-strong chain]
    Consider a validator $p$ and its local data structure $\text{ledger}^p$ (Algorithm~\ref{algo:fingadget}, line~\ref{ledgerStructure}). Let $S$ represent the set of distinct signatures from validators included in the valid certificates of blocks within the set of blocks in $\text{ledger}^p$ that were output by the SMR protocol $\Pi$ within the last $\Delta^*$. We say that $\text{ledger}^p$ is an $i$-strong chain, where $i = |S| - 2f - 1$.
\end{definition}


\vspace{3pt}
\noindent\textbf{Strong chain finalization rule.} Validators that monitor an $i$-strong chain for $i > f/4$, can finalize transaction volume up to $C = \frac{f}{f-i}D$. Otherwise, validators finalize $C = D$ coins. Notice that for $i=f$, $C \to \infty$ because all validators participate in that ledger, allowing the finalization of any block with a valid certificate. 

\[ C = 
\begin{cases} 
\frac{f}{f-i}D, & \text{if } i > \frac{f}{4}, \\
D, & \text{else.}
\end{cases}
\]

\vspace{3pt}
\noindent\textbf{Strongest chain finalization rule.}
We additionally define the strongest chain as an i-strong chain with \( i > \frac{f+1}{2} \). More than half of the correct validators participate in the strongest chain. This ensures that only one strongest chain can exist within \( \Delta^* \), enabling correct validators to finalize every transaction, under this rule,
 
\[
C = 
\begin{cases} 
\infty, & \text{if } i > \frac{f+1}{2}, \\
\frac{f}{f-i}D, & \text{if } \frac{f}{4} < i \leq \frac{f+1}{2}, \\
D, & \text{otherwise.}
\end{cases}
\]





\section{\sys: Recovery Mechanism}\label{sec:recovery}
\chris{updated (but needs a lot of  improvement on writing): i'm  also quite verbose now, for clarity}
In Section~\ref{sec:protocol}, we describe \sys's finalization rule which implements an $(n,k,f^*)$--resilient transformation for executions with $f^* \leq f$ Byzantine validators and $k \leq 2f - f^*$ rational validators. In this regime, although equivocations are possible with the participation of rational parties, they are not profitable for rational validators; thus, rational parties will not engage in forking attacks.

We now extend our protocol to executions with $f+1 \leq f^* < 2n/3$ Byzantine validators. In this setting, equivocations cannot be prevented, since $f+1$ Byzantine validators alone suffice to create blocks with conflicting certificates within $\Delta^*$. Accordingly, we no longer aim to preclude safety violations; instead, we provide a recoverable economic-security regime that guarantees economic restitution with recovery parameter $\Delta_R = O(f^*\Delta^*)$ (Definition~\ref{def:economicDef}).

\noindent At a high level, our recovery protocol works as follows:
\begin{itemize}[nosep,leftmargin=*]
    \item \textbf{Detecting conflicts:} Upon observing conflicting certificates, correct validators halt further finalization to prevent additional financial damage to clients. \sys's finalization rule guarantees that: i) the resulting economic damage is bounded by the slashable stake of the equivocating validators, ii) every block that may have been accepted by a client is observed by at least one honest validator, ensuring that all affected transactions can be incorporated into the recovery process. \chris{this is not really an execution phase of a protocol, we just that some properties hold} 
    \item \textbf{Set agreement:} Validators deterministically agree on a common set of blocks to execute, including all blocks that may have been accepted by clients during the equivocation (even the conflicting ones finalized within the last $2\Delta^*$). 
    \item \textbf{Execution and reimbursement:} Validators revert to a stable pre-fork state and deterministically re-execute all conflicting finalized ledgers in a common order, reimbursing provable client losses from the stake of provably misbehaving validators, without minting new coins.
    \item \textbf{New genesis:} Validators issue a verifiable commitment to the recovered state, allowing the system to resume from a single consistent ledger.
\end{itemize}

\subsection{Detecting conflicts}
We first bound the maximum economic damage that clients may incur, namely, the maximum value of transactions contained in blocks that clients may accept before an equivocation is detected and finalization is halted. Together with the fact that every such block is observed by at least one honest validator, these properties ensure that all affected clients can be reimbursed.

\vspace{3pt}
\noindent\textbf{Bounding economic damage.} Clients accept only blocks with \emph{finality certificates}. According to \sys's finalization rule, the total amount of coins carrying finality votes during a fork is at most $fD$. Once the synchronous bound elapses and correct validators are aware of the equivocation they stop finalizing new blocks. After observing the equivocation, all honest validators identify at least $f+1$ Byzantine validators which have signed conflicting finality votes (by quorum intersection). The stake of these validators is sufficient to cover and reimburse all clients affected by the fork.

\vspace{3pt}
\noindent\textbf{Observing all conflicting certificates accepted by clients.} 
Recall that each vote contributing to a finality certificate $\mathcal{F}(b)$ for some block $b$ is included in a subsequent block $b'$ carrying a valid SMR certificate $QC(b')$ (or is included in a subset of blocks each carrying a valid SMR certificate), which in our construction is broadcast to all validators. 
Nevertheless, malicious validators may withhold their votes for $QC(b')$ from other validators but send this certificate directly to clients. As a result, clients may accept $b$ even though honest validators never observe the certificate $QC(b')$ forming the finality certificate $\mathcal{F}(b)$. 

To ensure that every block $b$ accepted by clients is eventually observed by honest validators, we rely on the fact that the finality votes for $b$ become visible to an honest validator before the corresponding SMR certificate $QC(b')$ is formed. Indeed, any valid $QC(b')$ sent to clients contains at least one honest signature, implying that some honest validator has already observed the proposal for $b'$ by the corresponding leader, including the finality votes for $b$.

Motivated by this observation, we define an \emph{augmented finality certificate}, denoted $\hat{\mathcal{F}}(b)$, as the existence of $2f+1$ valid finality votes for $b$, \emph{regardless of whether these votes have been embedded in blocks carrying valid SMR certificates}.
For example, if a finality vote for block $b$ is included in a block $b'$ by a proposer, but $b'$ does not have a valid SMR certificate, the vote for $b$ counts for the augmented finality certificate $\hat{\mathcal{F}}(b)$ but not for the finality certificate $\mathcal{F}(b)$. 
Every block $b$ accepted by a client therefore carries either a finality certificate $\mathcal{F}(b)$ or an augmented finality certificate $\hat{\mathcal{F}}(b)$ observed by at least one honest validator $p$ before $t^*_p$, the time at which $p$ witnesses an equivocation. This ensures that all client-accepted transactions enter the recovery process and are eligible for reimbursement.

\subsection{Set Agreement}

So far, we have established two properties. First, the economic damage from an equivocation is bounded and can be reimbursed using the stake of the equivocating validators. Second, every block accepted by a client is eventually observed by at least one honest validator through a finality or augmented finality certificate. The remaining challenge is to ensure all honest validators agree on a common recovery set to execute and form the new genesis block. To this end, validators run a recovery protocol similar to~\cite{lewis2025beyond}, with the additional requirement of including every transaction accepted by clients. The set agreement protocol outputs a common DAG ledger $\mathcal{L}_R$ containing only blocks with a finality or augmented finality certificate, and necessarily all (even conflicting) blocks that have been accepted by clients.

During an initial dissemination phase, validators exchange the finality and augmented finality certificates they have observed\footnote{For ease of presentation, we say validators exchange all such certificates; the actual protocol communicates a bounded subset (Appendix~\ref{app:recovery}).}. By the end of this phase, every honest validator has learned every certificate observed by any honest validator at the time of equivocation. Since every client-accepted block is observed by at least one honest validator, all client-acceptable blocks are present in every honest validator's $p$ local recovery DAG $\mathcal{L}^p_R$.

Honest validators may nonetheless hold different local DAGs: in the dissemination phase, a Byzantine validator may selectively disseminate certificates to only a subset of honest validators. To reconcile these, the protocol runs a view-based set agreement procedure. In each view $u$, every honest validator $p$ sends its local DAG $\mathcal{L}^p_R$ to the leader $l^u$, who proposes the union of the received sets; validator $p$ votes only for proposals extending its own $\mathcal{L}^p_R$. Since conflicting finality certificates implicate at least $f+1$ Byzantine validators, fraud proofs eventually leave a reduced active set with an honest majority. Majority agreement then guarantees the adopted $\mathcal{L}_R$ contains the local DAG of at least one honest validator, and therefore all blocks that may have been accepted by clients. These blocks are executed and incorporated into the recovered state. 
We defer the full protocol specification to Appendix~\ref{app:recovery}.

\subsection{Execution and Reimbursement} The set agreement protocol outputs the DAG $\mathcal{L}_R$ to every honest validator. We now describe how validators execute transactions, including those in the conflicting ledgers $\mathcal{L}$ of the DAG $\mathcal{L}_R$, to update the system state.

Recovering a common state requires that all validators execute the same transactions in the same order, even in the presence of conflicting finalized ledgers. 
To this end, validators revert to the state up to a block for which no conflicting finalized block exists. To support this efficiently, validators locally maintain two distinct states. 
\begin{itemize}[nosep,leftmargin=*]
    \item \emph{Recent-finalized state:} the state up to their most recent finalized block,
     \item \emph{Old-finalized state:} the state up to their most recent $2\Delta^*$--old finalized block (i.e., output $2\Delta^*$ ago).
 \end{itemize}

Upon detecting an equivocation, validators revert to their old-finalized state (line~\ref{alg:line:execute1a}), for which it is guaranteed that no fork exists. After the set agreement protocol outputs the DAG $\mathcal{L}_R$ (line~\ref{alg:line:setAgreement}),  they also execute transactions in conflicting blocks across $\mathcal{L}_R$.

To do so, validators first re-execute all transactions from their old-finalized state up to the first point of conflict, namely the first block in $\mathcal{L} = \{L_1, \ldots, L_j\}$ (lines \ref{alg:line:execute1a}-\ref{alg:line:execute1b}). For the conflicting blocks, they sort the conflicting ledgers in $\mathcal{L}$ lexicographically by the hash of their last finalized block, forming a new ordered set $\mathcal{L}^*$ (line~\ref{alg:sortL}). Finally, validators execute each ledger in $\mathcal{L}^*$ sequentially, processing all block within each ledger in order (lines \ref{alg:line:execute2a}-\ref{alg:line:execute2b}).

During execution, transactions fall into two categories.
First, transactions issued by honest clients who believed a branch to be canonical remain valid, since honest clients do not spend more coins that they hold. Second, transactions from malicious participants attempting to double-spend. For these, validators reimburse affected clients from the misbehaving validators' slashed stake and burn any remaining stake to punish the malicious validators (line~\ref{alg:line:Slash}).

\begin{algorithm}[t!]
\scriptsize
\setstretch{0.95}
\centering
\caption{\textsl{\sys protocol}}
\label{algo:fingadget}
\begin{algorithmic}[1]
\State{$ledger \gets \emptyset$} \label{ledgerStructure}\Comment{Blocks with valid certificates}
\State{$quorumstore \gets \emptyset$} \Comment{Valid certificates}
\State{$final\_votes \gets \emptyset, final\_store \gets \emptyset $} \Comment{Finality votes/certificates}
\State{$suffixindex \gets 0$} \Comment{$\Delta^*$-recent blocks}
\State{$oldFinalized, oldFinalizedBlock \gets G, \bot$} \Comment{State/block from $2\Delta^*$ ago}
\State{$ finalized \gets G$} \Comment{Recent finalized state}
\State{$violation \gets false$} \Comment{Safety violation flag}
\State{$blacklist \gets \emptyset$} \Comment{Blacklisted validators}

\Upon{\textsc{SMR.deliver}($QC(b), b, h$)} \label{alg:line:consensusoutput}
\If{$quorumstore[h] = \bot$} \label{alg:line:checkempty}
\State{$quorumstore[h] \gets QC(b)$} \label{alg:line:quorum}
\State{$ledger \gets ledger \cup (b, localtime)$} \label{alg:line:ledger}
\State{\textsc{finalize-blocks}()}
\Else
\State{$(block, \_\_ \;) \leftarrow ledger[h]$}
\State{\textbf{if } $block \neq b$ \textbf{then} $violation \gets true$}\label{alg:line:violation1}
\EndIf
\EndUpon

\Upon{\textsc{observing (``finality", $u, b$)} in a phase--$(s-1)$ certificate}
    \State{$final\_votes[b].add(u)$}
    \If{$|final\_votes[b]| \geq 2f+1 \land final\_store[b.height] = \bot$}
    \State{$finalized.apply(b) $ \Comment{Execute $b$ to update the state}}
    \State{$(\_\_, time) \gets ledger[b.height]$}
    \If{$localtime - time \geq 2\Delta^* \land  b.height > oldFinalizedBlock.height $ }
        \State{$ oldFinalized.apply(b)$, $oldFinalizedBlock \gets b$}
    \EndIf
    \ElsIf{$|final\_votes[b]| \geq 2f+1 \land final\_store[b.height] \neq \bot $ }
    \State{$ (b', \_\_) \leftarrow final\_store[b.height]$}
    \If{$b' \neq b $}
            \State{$violation \gets true$} \label{alg:line:violation2}
            \State{$blacklist \gets blacklist \cup (final\_store[b.height] \cap final\_votes[b])$}\label{alg:line:blacklist}
           \State{\textbf{if} it is the first time of equivocation run, \textsc{Recover}}
            \State{$final\_store[b.height].add(final\_votes[b])$}
        \EndIf
    \EndIf
\EndUpon

\Procedure{\textsc{finalize-blocks}}{} \label{alg:line:finalizeblocks}
\State{$suffixvalue = 0$}
\For{$(block, blocktime) \in ledger$ \textbf{from} suffixindex \textbf{to} $|ledger|$} \label{alg:line:iterateledger}
\If{$blocktime \leq localtime - \Delta^* \land \;  (not \;  violation) $} \label{alg:line:deltafinal}
\State{$block.finalize$} \label{alg:line:finalizeD}
\State{\textbf{broadcast} signed finality vote $v$: (\texttt{"finality"}, $v$, block)} \label{alg:line:finality1}
\State{$suffixindex++$}
\ElsIf{$suffixvalue + block.value \leq C \land \;  (not \;  violation) $} \label{alg:line:collateralfinal}
\State{$block.finalize$} \label{alg:line:finalizeC}
\State \textbf{broadcast} signed finality vote $v$: (\texttt{"finality"}, $v$, block) \label{alg:line:finality2}
\State{$suffixvalue \gets suffixvalue + block.value$}
\Else
\State{\textbf{return}}
\EndIf
\EndFor
\EndProcedure

\Procedure{\textsc{recover}}{}
    \State $ \mathcal{L}_R \leftarrow \textsc{set-agreement (see Algorithm~\ref{alg:setagreement})}$ \label{alg:line:setAgreement}
    \State{$\mathcal{L} \gets \{L_1, ..., L_j\}$} \Comment{conflicting finalized ledgers in $\mathcal{L}_R$}\label{alg:defL}
    \State $\mathcal{L}^{*} \gets$ $Sort(\mathcal{L})$ \Comment{sort ledgers by hash of last block} \label{alg:sortL}
    \State{$recovered\_state \gets oldFinalizedBlock$} \label{alg:line:execute1a}
\For{$pos$ \textbf{from} $oldFinalizedBlock.height+1$ \textbf{to} $|ledger|$}
   \State  \textbf{if} $ledger[pos] \in L$ for some $L \in \mathcal{L}$ \textbf{continue}
    \State $recovered\_state.apply(ledger[pos])$ \label{alg:line:execute1b}
\EndFor
\ForAll{ledger $L \in \mathcal{L}$}\label{alg:line:execute2a}
    \ForAll{block $b \in L$}
        \State $recovered\_state.apply(b$) \label{alg:line:execute2b}
    \EndFor
\EndFor
    \State{$recovered\_state \gets$ \textsc{slash}($ blacklist , \; recovered\_state$)} \label{alg:line:Slash}
    \State \textbf{broadcast} signed state commitment to $recovered\_state$ \label{alg:line:broadcastCmt}
\EndProcedure
\end{algorithmic}
\end{algorithm}

\subsection{New Genesis}
After reaching a common state, validators produce a signed state commitment (e.g., a Merkle root) for the recovered state, that can be verified by new participants. 
To ensure that only honest validators can produce this certificate, it includes equivocation evidence identifying at least $f+1$ validators from the original set $N$ that are proven to have equivocated (i.e., by issuing conflicting finality votes). A state commitment is considered a \emph{genesis block} if it contains: (i) equivocation proofs for at least $f+1$ distinct validators in $N$, and (ii) signatures from at least $f+1$ validators that are not among those proven to have equivocated.

\section{Security Guarantees}
In this section, we provide sketch proofs and defer the formal analysis to 
\ifsubmissionversion
Appendix~F in the extended version~\cite{fullversion}.
\else
Appendix~\ref{sec:Secanalysis}.
\fi

\vspace{3pt}

\noindent\textbf{Safety and client-safety.} Rational validators may collude with Byzantine validators to create conflicting blocks within the maximum message delay $\Delta^*$. However, such attacks are unprofitable: to double-spend, misbehaving validators must present inclusion proofs for blocks with finality certificates, ensuring they are finalized by at least one correct validator. During a fork, at most $fD$ coins can be double-spent in total, and any single rational validator gains at most $D$ coins, for all finalization rules (stake-bounded, strong, or strongest chain), by \ifsubmissionversion
Lemmas~2,4,6 (extended~\cite{fullversion}) \else 
Lemmas~\ref{lemma:bound1},~\ref{lemma:collateralStrongBlocks}, and~\ref{lemma:collateralStrongestBlock}\fi, respectively. Every correct validator observes all equivocations for blocks with finality certificates within \(2\Delta^*\) and blacklists misbehaving validators via our slashing mechanism. Hence, misbehaving validators cannot withdraw their stake 
\ifsubmissionversion
(Lemma~3, \cite{fullversion}). \else 
(Lemma~\ref{lemma:successfulslashing}).
\fi


Finally, we conclude that rational validators have no incentive to engage in forking attacks, as any potential gain does not exceed their stake, which is subject to slashing. This ensures the protocol's safety 
\ifsubmissionversion
(Theorem F.1, \cite{fullversion}). \else 
(Theorem~\ref{theorem:totalordering}).
\fi 
Moreover, as clients accept inclusion proofs only for blocks with at least \( 2f+1 \) finality votes, client-safety is guaranteed 
\ifsubmissionversion
(Theorem F.2, \cite{fullversion}). \else 
(Theorem~\ref{theorem:certifiable}).
\fi


\vspace{3pt}
\noindent\textbf{Liveness and client-liveness.}
Since transaction fees are distributed among participating validators, rational validators have no incentive to withhold proposals or votes for valid blocks. However, Byzantine validators may attempt to bribe them to silence correct proposers and censor transactions. 

Byzantine validators claim rewards only within windows where at least \( 2f+1 \) unique proposers exist, ensuring at least one is correct. Once their rewards are depleted, rational validators resume proposing and voting for valid blocks 
\ifsubmissionversion
(Lemma 10, \cite{fullversion}). \else 
(Theorem~\ref{lemma:liveness}).
\fi 

Consequently, a correct validator will eventually propose a block, ensuring liveness \ifsubmissionversion
(Theorem F.1, \cite{fullversion}). \else 
(Theorem~\ref{theorem:totalordering}).
\fi  Additionally, to claim their rewards, rational validators commit finality votes, guaranteeing client-liveness \ifsubmissionversion
(Theorem F.2, \cite{fullversion}). \else 
(Theorem~\ref{theorem:certifiable}).
\fi 

\chris{To do: refer to lemmas when they are ready}
\vspace{3pt}
\noindent\textbf{Recovery and economic restitution.} After a fork, all correct validators observe the equivocation and, every block carrying (augmented) finality certificates, that could have been accepted by clients, is observed by at least one honest validator. By \sys’s finalization rules, the total value of coins that can be double-spent across conflicting finalized blocks is bounded by \(fD\) \ifsubmissionversion
(Lemmas~2,4,6, extended version~\cite{fullversion}). \else 
(Lemmas~\ref{lemma:bound1},~\ref{lemma:collateralStrongBlocks}, and~\ref{lemma:collateralStrongestBlock}).
\fi
Validators provably identify at least \(f+1\) misbehaving validators, whose combined stake \((f+1)D\) covers all double-spent transactions. During the set agreement phase, validators agree on the same set of blocks to execute for the new genesis, ensuring deterministic recovery 
\ifsubmissionversion
(Lemma 16, \cite{fullversion}) \else 
(Lemma~\ref{lemma:commonRecoveryLedger})\fi
; moreover, all blocks with finality or (augmented) finality certificates accepted by some client is included in the recovery set \ifsubmissionversion
(Lemma 16, \cite{fullversion}). \else 
(Lemma~\ref{lemma:commonRecoveryLedger}).\fi Finally, they deterministically execute the transactions from conflicting blocks in the same order to agree on a common state and reimburse affected clients using the slashed stake 
\ifsubmissionversion
(Theorem F.3, \cite{fullversion}). \else 
(Theorem~\ref{th:recovery}).
\fi 

\section{Case-Study}
\label{sec:case-study}

We evaluate \sys's practical feasibility on one million consecutive Cosmos blocks. Over 99\% of these blocks exceeded the $\frac{5}{6}$ participation threshold required by our strongest-chain finalization rule, allowing \sys to instantly finalize the vast majority of blocks\footnote{Source at \url{https://anonymous.4open.science/r/cosmos-analytics-4994/}.}.

The remaining blocks were produced during periods of lower participation, lasting up to 147 blocks. Given Cosmos’s $\sim$6-second block time, this corresponds to $\sim$15 minutes. With a median daily transaction volume of \$17 million, this 15-minute window corresponds to an average transaction volume of $\sim$\$180{,}000.
As such, given the minimum slashable stake of active validators of $\sim$\$470{,}000, \sys can compensate for low-participation periods finalizing blocks using our stake-based finalization rule.

These high participation rates are not unique to Cosmos. On Ethereum, the average daily participation has also consistently remained above 99\%~\cite{beaconscan}.
Thus, we conclude that \sys could be seamlessly integrated into production systems like Cosmos or Ethereum, enhancing system resilience without impacting block confirmation times.

\section{Limitations and Extensions}
\label{sec:discussion}

 \ifsubmissionversion
 We discuss next the limitations and extensions of our protocol, while a detailed discussion on the network model and incentive mechanisms is deferred to Appendix~\ref{discussion:appendix}.
 \else 
\fi

\vspace{3pt}
\noindent\textbf{Dynamic validator settings.} We outline how to extend our protocol to the permissionless setting, where validators dynamically join or leave, following dynamic-membership BFT~\cite{duan2022foundations}. As in~\cite{duan2022foundations}, we assume every membership change preserves our resilience assumptions: at any time at least $1/3$ of active validators are honest, and between consecutive membership changes a quorum of at least $2/3$ honest and rational validators remains stable.

\vspace{3pt}
\noindent\emph{On-chain join requests.} Like withdrawals, we treat join requests as on-chain state events: once a request is committed, validators consistently update their local view, so they all agree on the current validator set and quorum size.


\vspace{3pt}
\noindent\emph{Challenges.} To participate safely, a validator joining later must: (i) learn the current validator set, (ii) learn the system state, and (iii) detect past safety violations, so that it neither lets malicious validators withdraw their stake nor unknowingly extends a fork. For (i) and (ii), late joiners obtain the full chain, execute the blocks to reconstruct the state, and replay membership-change events from genesis to obtain the current validator set; periodic checkpoints for the validator set and state commitments are left as future work.


\vspace{3pt}
\noindent\emph{Past equivocations.}
To ensure that late joiners do not mistakenly assist misbehaving validators, they must also learn whether any equivocation occurred in the past. We use the following mechanism:

\begin{itemize}[nosep,leftmargin=*]
    \item \emph{Late joiner}. A validator wishing to join at time~$t$ broadcasts a join request and waits until time~$t + 2\Delta^*$ before participating fully. This waiting period ensures that it can receive any messages revealing past safety violations.

    \item \emph{Existing validator}. Any existing validator $p$ that receives the join request at some time $t' \in \{t, t{+}\Delta^*\}$ responds to the joiner with all safety violations it is aware of up to time $t'$. After that moment, $p$ treats the joining validator as active and interacts with it according to the protocol.
\end{itemize}

By time $t' + \Delta^* \leq t+ 2\Delta^*$, the joining validator will have learned all equivocation evidence known to validator $p$ by time $\le t'$, and will learn about future inconsistencies during normal protocol participation. 

\vspace{3pt}
\noindent\textbf{Assessing transaction value.}
An important part of our protocol is assessing the value of transactions in a given block to weigh it against the locked stake. While this is simple in a payment system environment, it is not straightforward in a smart contract environment. While transactions transferring native on-chain currency can easily be evaluated, it is challenging to evaluate the value of ERC-20 tokens, NFTs, or even layer-2 rollup checkpoints. To address this, clients can include an estimated transaction value in their submissions. If the receiver agrees with the estimate, they accept the stake-based finality rule; otherwise, they wait for the full finality threshold \( 2\Delta^* \) (or a strong chain). 


\vspace{3pt}

\noindent\textbf{Generalized rewards scheme.}
The reward mechanism requires that every validator appears as a proposer at least once in each window $W_k$ (with high probability). The window length $w$ depends on the protocol's leader-election mechanism. For protocols where each validator is independently and uniformly selected as leader in each round, choosing $w \ge n(\log n + \lambda)$ ensures that every validator is elected at least once with probability at least $1-e^{-\lambda}$, where $\lambda$ is a security parameter, while, for leader-election probabilities $\{q_i\}_{i=1}^n$ \cite{gilad2017algorand}, choosing $w \ge (\ln n+\lambda)/q_{\min}$, where $q_{\min}=\min_i q_i$, yields the same guarantee (appendix~\ref{appendix:GeneralizedReward}). For multi-proposer protocols where progress requires successful proposals from at least $2f+1$ validators in each round \cite{babel2023mysticeti}, every successful round already guarantees participation from a quorum of validators, reducing the window to $w=1$.
 \chris{citations, discuss with someone knowing crypto to see what applies here
}

\chris{citations}

\ifsubmissionversion
\else 
\vspace{3pt}
\noindent\textbf{{Bribing mechanism.}}
Trustless bribing mechanisms enable cooperation between Byzantine and rational validators. A Byzantine validator can bribe a rational validator \( p \) to sign two conflicting blocks \( b \) and \( b' \) (e.g., by sending pre-signed transactions). Similarly, censorship of a target client \( p' \) can be incentivized, effectively rewarding exclusion of \( p' \)'s transactions. In particular, the Byzantine validator can deploy a smart contract, funded every second block, that pays out in the next block if \( p' \)'s wallet balance remains unchanged, allowing rational validators to claim the bribe only if \( p' \)'s wallet is not updated.

\vspace{-0.005cm}
\noindent{\textbf{External incentives.}} External incentives or attacks that are not observable within the blockchain ecosystem are incorporated by Byzantine behavior. By modeling a rational–Byzantine adversary, we capture a system that is resilient to both external  incentives (Byzantine behavior), and internal incentives (rational behavior).

\vspace{-0.005cm}
\noindent{\textbf{$(\Delta, \Delta^*)$--timing model.}}
Another interpretation of our synchrony assumption is through the $(\Delta,\Delta^*)$-timing model introduced in prior work~\cite{budish2024economic,lewis2025beyond}. In this model, messages are delivered within a known bound $\Delta$ after GST, while an additional (potentially much larger) bound $\Delta^*$ is assumed to hold at all times, including before GST. The bound $\Delta^*$ can be viewed as representing slower but reliable communication channels that remain available during network disruptions or adversarial conditions. As shown in Section~\ref{sec:impossibilitiesSummary}, rational resilience is impossible when Byzantine and rational validators together exceed one third of the system ($f^ + k > n/3$). Consequently, some synchrony assumption is necessary to obtain a solution, and we therefore work within the $(\Delta,\Delta^*)$-timing model.

\noindent{\textbf{Probabilistic network model.}} In line with prior work, we assume an adversary that controls both malicious validators and network delays. We leave for future work models in which control over validators and the network is decoupled~\cite{danezis2025byzantine} or network partitions are probabilistic~\cite{guo2019synchronous}. In such settings, attacks would succeed only with some probability, requiring the protocol to ensure that the expected utility of any attack for rational validators is strictly less than the stake they forfeit. 
\fi

\section{Related Work}
\chris{check who the reviewers are :)}

\com{
\subsubsection*{\textbf{Incentive Compatibility}} Blockchains usually operate in financial contexts, and thus, their validators which are responsible for their maintenance and security, are typically financially motivated. Consequently, designing incentive-compatible consensus protocols is crucial, particularly given the permissionless nature of blockchains. While production blockchains such as Bitcoin and Ethereum incorporate incentive mechanisms, they are not incentive-compatible~\cite{eyal2015miner,eyal2018majority, sapirshtein2017optimal, tsabary2018gap, wood2014ethereum}.  

In the PoW setting, the state-of-the-art incentive-compatible blockchain protocol is FruitChains~\cite{pass2017fruitchains}, which remains secure as long as rational miners do not form a coalition that exceeds $1/2$ of the total mining power.
Several subsequent works have proposed diverse incentives mechanisms for PoW blockchains to optimize specific metrics, e.g., latency~\cite{zhang2019lay}. 

While PoW blockchains continue to evolve, environmental concerns and scalability limitations have driven the blockchain community towards Proof of Stake (PoS) blockchains. Recent research has particularly focused on accountable PoS systems~\cite{polygraph,abc,mixedrational4,lewispye2025accountableliveness,aftbyzratpaper}, where misbehaving participants can be identified and penalized.
}

\vspace{-3pt}
\noindent\textbf{Rational-Byzantine model in quorum-based protocol.}
BAR~\cite{bar} studies SMR with altruistic (honest), rational, and Byzantine parties. However, BAR is designed for general-purpose state machine replication and does not address the economic incentive structures or adversarial dynamics specific to payment systems. Moreover, BAR assumes that rational participants act independently and do not collude. Plenty of works considering rational and Byzantine parties also do not address collusions~\cite{amoussou2019rationals,mcmenamin2021achieving,lev2019fairledger}. 

Abraham et al.~\cite{abraham2006distributed,abraham2008lower} introduce a solution concept  which allows strategic collusion between rational and Byzantine players in the context of secret sharing and multi-party computation. In this model, TRAP~\cite{trap} solves one-shot consensus in the partially synchronous setting for 
\( n > \max \left( 3/2k + 3f, 2(k + f) \right) \). However, we show that this approach does not extend to multi-shot SMR in partial synchrony, as the total value that can be double-spent cannot be bounded. 
Assuming a synchronous network, we instead present an SMR protocol that tolerates up to \( f + k < 2/3 \), improving upon TRAP’s resilience bound of \( f + k < 1/2 \).

\vspace{3pt}
\noindent\textbf{Accountability.}
Several works focus on accountability in distributed protocols, including specific constructions~\cite{polygraph, buterin2017casper,shamis2022ia}, analyzing forensic properties of well-known protocols~\cite{sheng2021bft}, and generic transformations that make protocols accountable~\cite{abc}. While these approaches ensure that misbehaving validators can eventually be detected, accountability alone does not guarantee rational resilience. If rational validators profit from double-spending is higher than their stake, they will misbehave even if eventually their stake is slashed. Moreover, this line of work does not address recovery after safety violations or client reimbursement. Our work closes these gaps by (i) identifying when accountability suffices to deter rational attacks and (ii) providing a recovery mechanism that ensures full reimbursement to clients when the Byzantine threshold is exceeded.

\vspace{3pt}
\noindent\textbf{Economic safety.} 
BLR24~\cite{roughgarden} examines economic safety in PoS blockchains, focusing on when slashing mechanisms can reliably punish equivocation, while ensuring honest validators retain their stake. They prove that slashing is impossible when more than $2f$ out of $3f+1$ validators misbehave, and propose a slashing mechanism when at most $2f$ do. 
Although our bounds align, as we also provide resilience up to $2f$  Byzantine and rational nodes, our focus is fundamentally different. 
BLR24~\cite{roughgarden} does not account for liquid coin transfers during SMR execution or explicitly model profit generation beyond slashing deterrents.

\vspace{3pt}
\noindent\textbf{Recovery.} Lewis-Pye et al.~\cite{lewis2025beyond} study Byzantine adversaries controlling more than one third of the validators in quorum-based SMR. They prove that in partial synchrony, recovery protocols cannot guarantee a bounded rollback window, so an unbounded number of finalized transactions may need to be reverted. Consequently, such protocols~\cite{zlb, gong2025recover} cannot provide full client reimbursement and are orthogonal to our approach, which targets client-level economic security. Moreover, assuming a synchronous bound $\Delta^*$, \cite{lewis2025beyond} proposes a generic recovery protocol tolerating up to $2/3$ Byzantine failures, in which each instance completes in $O(f\cdot \Delta^*)$ rounds and rolls back transactions finalized within the last $2\Delta^*$, but does not compensate clients that have already accepted those transactions. Under the same network assumptions, \sys adds client-level economic security: all client-accepted transactions are incorporated into recovery and added to the new genesis block, while losses from conflicting transactions accepted during the final $2\Delta^*$ are fully reimbursed from the stake of provably misbehaving validators, by combining a variant of~\cite{lewis2025beyond} with $\sys$'s confirmation rule.

Sridhar et al.~\cite{sridhar2023bettersafesorryrecovering} describe a \emph{gadget} layered onto synchronous SMR protocols so that clients do not accept conflicting ledgers, but only if clients actively communicate and wait $3\Delta$ before finalizing. In contrast, our clients remain silent, accepting any block that carries finality votes, which can be finalized immediately. Moreover, \cite{sridhar2023bettersafesorryrecovering} does not specify recovery after safety violations, assuming instead that an honest majority is restored externally.

\chris{check}
\vspace{3pt}
\noindent\textbf{Synchronous BFT.}
Synchronous BFT SMR protocols~\cite{abraham2020sync} leverage synchrony to remain secure under an honest majority assumption. To achieve this, the decision of each value fundamentally depends on an explicit synchrony bound. In contrast, in line with \cite{lewis2025beyond}, our focus is to enable a transformation of partially synchronous protocols so that they operate within their assumed model without relying on a known synchrony bound. Even in the presence of fewer than a supermajority of validators, \sys finalizes blocks up to the stake threshold (and, under high participation, potentially all blocks) before the synchronous bound $\Delta^*$ is reached. 

 \ifsubmissionversion
Due to space constraints, we discuss 
flexible BFT paradigms, and related impossibility results in Appendix~\ref{app:relatedWork}. \else 
\fi

\noindent\textbf{Flexible BFT.} Decoupling the confirmation rules of clients and validators originates from the flexible BFT paradigm~\cite{malkhi2019flexible, xiang2021strengthened, neu2024optimal} which separates the role of proposing and voting transactions (done by validators) with committing transactions (done by clients). We implement this separation using finality certificates for transactions and also leveraging certificates with high participation in our strong block finality rules. 

\noindent\textbf{Impossibilities.} Finally, our impossibility results in partial synchrony align with~\cite{tas2023bitcoin,roughgarden}, which show that more than $1/3$ Byzantine validators can cause safety violations and avoid slashing once the security threshold is exceeded. We extend this by showing that the impossibility holds for any combination of Byzantine and rational participants exceeding $1/3$. In this setting, rational validators cannot be disincentivized from joining the attack.

 \ifsubmissionversion
 \else 
 \section{Acknowledgments}
The work was partially supported by the Austrian Science Fund (FWF) through the SFB SpyCode project F8512-N, and by the WWTF through the projects 10.47379/ICT22045 and 10.47379/ICT25056. Finally, we
thank Pim Keer and Ioannis Alexopoulos for insightful discussions.
\fi




\bibliographystyle{IEEEtran}
\bibliography{bib}

\appendices

 \ifsubmissionversion
 \section{Further Related Work}\label{app:relatedWork}
\noindent\textbf{Flexible BFT.} Decoupling the confirmation rules of clients and validators originates from the flexible BFT paradigm~\cite{malkhi2019flexible, xiang2021strengthened, neu2024optimal} which separates the role of proposing and voting transactions (done by validators) with committing transactions (done by clients). We implement this separation using finality certificates for transactions and also leveraging certificates with high participation in our strong block finality rules. 

\noindent\textbf{Impossibilities.} Finally, our impossibility results in partial synchrony align with~\cite{tas2023bitcoin,roughgarden}, which show that more than $1/3$ Byzantine validators can cause safety violations and avoid slashing once the security threshold is exceeded. We extend this by showing that the impossibility holds for any combination of Byzantine and rational participants exceeding $1/3$. In this setting, rational validators cannot be disincentivized from joining the attack.
\section{Discussion}\label{discussion:appendix}
\vspace{-0.15cm}
\noindent\textbf{{Bribing mechanism.}}
Trustless bribing mechanisms enable cooperation between Byzantine and rational validators. A Byzantine validator can bribe a rational validator \( p \) to sign two conflicting blocks \( b \) and \( b' \) (e.g., by sending pre-signed transactions). Similarly, censorship of a target client \( p' \) can be incentivized, effectively rewarding exclusion of \( p' \)'s transactions. In particular, the Byzantine validator can deploy a smart contract, funded every second block, that pays out in the next block if \( p' \)'s wallet balance remains unchanged, allowing rational validators to claim the bribe only if \( p' \)'s wallet is not updated.

\vspace{-0.005cm}
\noindent{\textbf{External incentives.}} External incentives or attacks that are not observable within the blockchain ecosystem are incorporated by Byzantine behavior. By modeling a rational–Byzantine adversary, we capture a system that is resilient to both external  incentives (Byzantine behavior), and internal incentives (rational behavior).

\vspace{-0.005cm}
\noindent{\textbf{$(\Delta, \Delta^*)$--timing model.}}
Another interpretation of our synchrony assumption is through the $(\Delta,\Delta^*)$-timing model introduced in prior work~\cite{budish2024economic,lewis2025beyond}. In this model, messages are delivered within a known bound $\Delta$ after GST, while an additional (potentially much larger) bound $\Delta^*$ is assumed to hold at all times, including before GST. The bound $\Delta^*$ can be viewed as representing slower but reliable communication channels that remain available during network disruptions or adversarial conditions. As shown in Section~\ref{sec:impossibilitiesSummary}, rational resilience is impossible when Byzantine and rational validators together exceed one third of the system ($f^ + k > n/3$). Consequently, some synchrony assumption is necessary to obtain a solution, and we therefore work within the $(\Delta,\Delta^*)$-timing model.

\noindent{\textbf{Probabilistic network model.}} In line with prior work, we assume an adversary that controls both malicious validators and network delays. We leave for future work models in which control over validators and the network is decoupled~\cite{danezis2025byzantine} or network partitions are probabilistic~\cite{guo2019synchronous}. In such settings, attacks would succeed only with some probability, requiring the protocol to ensure that the expected utility of any attack for rational validators is strictly less than the stake they forfeit. 
 \else 
\fi

\section{Generalization to Randomized Leader Election}
\label{appendix:GeneralizedReward}
In this section, we discuss how the leader window $w$ can be chosen for protocols that employ randomized leader election. First, we consider protocols that elect leaders independently and uniformly at random\chris{citations, check what protocols give that}, and thus, each validator is elected leader with probability $1/n$ in every round. We will pick a value of $w$ such that every validator is elected at least once during the leader window with high probability is some security parameter $\lambda$. 

Let $A$ denote this event. Let also $E_i$ denote the event that validator $p_i$ is never elected during some window of $w$ rounds. Since $p_i$ is elected with probability $1/n$ in each round, we have $\Pr[E_i] = \left(1-\frac{1}{n}\right)^w$. Now, let $E=\bigcup_{i=1}^{n} E_i$ denote the event that at least one validator is never elected during the leader window. By the union bound: $$\Pr[E]
\le
\sum_{i=1}^{n}\Pr[E_i] \leq n\left(1-\frac{1}{n}\right)^w \leq ne^{-w/n} $$

Since event $A$ is the complement of $E$, we have that: 
$ \Pr[A]
= 1-\Pr[E]
\ge
1-ne^{-w/n}$. Now, by choosing $w=n(\ln n+\lambda)$, for some security parameter $\lambda>0$: $ \Pr[A]
= 1-\Pr[E]
\ge
1-e^{-\lambda}$. 

Thus, a leader window of size $w=n(\ln n+\lambda)$ guarantees that every validator is elected at least once with probability at least $1-e^{-\lambda}$. It is easy to see that this analysis generalizes to protocols that elect $k\ge 1$ leaders independently and uniformly at random in each round. In this case, choosing $w \ge \frac{n}{k}(\ln n+\lambda)$ gives us the same probability for event $A$.


We now extend the analysis to protocols that employ weighted randomized leader election, such as VRF-based stake-weighted leader election~\cite{gilad2017algorand}. In particular, we assume that validator $p_i$ is elected leader independently in each round with probability $q_i$, where $\sum_{i=1}^{n} q_i = 1$. Also, let $q_{\min}=\min_{i\in[n]} q_i$.

As before, let $A$ denote the event that every validator is elected at least once during a leader window of $w$ rounds, and let $E_i$ denote the event that validator $p_i$ is never elected during that window.  Then, $\Pr[E_i]
=(1-q_i)^w
\le
e^{-q_i w}$.

Let again $E=\bigcup_{i=1}^{n} E_i$ denote the event that at least one validator is never elected during the leader window and let $A$ be the complement of $E$.  By the union bound,
$\Pr[E]
\le
\sum_{i=1}^{n} e^{-q_i w} \le
n e^{-q_{\min} w}$. Therefore, by choosing $w=\frac{\ln n+\lambda}{q_{\min}}$ we obtain:  $\Pr[A]
\ge
1-e^{-\lambda}$.

\section{Set Agreement Protocol}\label{app:recovery} \chris{polish}
We now present the set agreement protocol, which builds upon~\cite{lewis2025beyond}, and outputs a DAG of blocks $\mathcal{L}^R$ to every honest validator. Our protocol provides an additional block-inclusion guarantee: the DAG of blocks $\mathcal{L}^R$ includes all blocks with a finality certificate or an augmented finality certificate, necessarily including all blocks accepted by clients.

Following an equivocation, honest validators may observe different sets of certificates and therefore may hold different views of the blocks that should be taken into account for the recovered ledger. For that reason, we will first characterize the local information available to validators at the start of recovery and, then, re-examine the objective of recovery with respect to the different local information of validators. 

\chris{changed notation for recovery inputs}
\noindent\textbf{Initial recovery inputs.}
Let $t^*$ denote the first time at which an equivocation is observed by some honest validator and let $t_p^*$ denote the time at which validator $p$ first detects the equivocation. Furthermore, let $L_p^{<t_p^*}$ denote the ledger consisting of every block $b$ with finality certificate $\mathcal{F}(b)$ observed before time $t_p^*$ by $p$, and let $S_R^p$ be the set of all blocks for which $p$ observes a finality certificate $\mathcal{F}(b)$ during the interval
$[t_p^*, t_p^*+2\Delta^*]$ along with all the blocks for which $p$ has observed only an augmented finality certificate $\hat{\mathcal{F}(b)}$.

Observe that for every block $b$ accepted by a client, it follows that, either: i) $b$ is included in $L_p^{<t_p^*}$ of every honest validator $p$ through a finality certificate, or ii) there is at least one honest validator $p$, for which, a certificate $\mathcal{F}(b)$ or  $\hat{\mathcal{F}}(b)$ for $b$ is included in $S_R^p$. Similarly, observe that for every block $b$ in $L_p^{<t_p^*}$, and any honest validator $p'$, $b$ is included either in $L_p^{<t_{p'}^*}$ or is included in $S_R^{p'}$. 

To ensure that all honest validators begin recovery with inputs containing (i) every block that may have been accepted by a client and (ii) every block for which a finality certificate was observed by some honest validator during the $2\Delta^*$ waiting period, we introduce an additional dissemination phase, leading to validators' extended recovery inputs. 

\vspace{-0.2cm}
\subsubsection*{\textbf{Extended recovery inputs}} Each validator $p$ broadcasts its initial recovery input $S_R^p$ at time $t^*_p + 2\Delta^*$. Additionally, during the interval
$[t_p^*+2\Delta^*, t_p^*+4\Delta^*]$, validator $p$ collects the initial recovery inputs from other validators and constructs the set $\hat{S_R^{p}} = \bigcup_{p' \in P} S_R^{p'}$ where $P$ is the set of validators that sent their initial recovery input during this period. We refer to $\hat{S_R^{p}}$ as the extended recovery input of $p$.

Since every honest validator detects the equivocation within $\Delta^*$ of its occurrence, every honest validator $p'$ sends their input $S_R^{p'}$ by time $t_{p'}^* +2\Delta^* \le t_p^*+3\Delta^*$, and thus all honest validators are included in $P$. However, we stress that honest validators do not necessarily have the same extended recovery input. That is because a Byzantine validator $p_m$  may selectively disseminate its recovery input $S_R^{p_m}$  to only a subset of honest validators. Consequently, some honest validators may include additional certificates that are not present in the extended recovery inputs of others.

\noindent\textbf{Active validator set}.
Each validator \(p\) maintains an active validator set \(A^p\), consisting of all validators for which \(p\) has not observed a valid equivocation proof. Let \(C_{A^p}\) denote the collection of equivocation proofs for validators excluded from \(A^p\). These proofs serve as justification for the active validator set \(A^p\) and are constructed from conflicting certificates in the extended recovery set of $p$.

Observe that every active validator set contains all \(f+1\) honest validators, since honest validators never equivocate and therefore cannot be excluded by a valid equivocation proof. Moreover, at least \(f+1\) Byzantine validators have provably equivocated and are excluded from every active validator set and, thus, every active validator set contains at most \(2f\) validators. Finally, every active validator set has an honest majority. To see this, consider an active validator set \(A^p\) with \(|A^p| = 2f-i\) validators, for some \(i \geq 0\). Since all \(f+1\) honest validators belong to \(A^p\), the set contains at most \(f-i-1\) Byzantine validators. Therefore, the number of Byzantine validators is strictly smaller than \(|A^p|/2\), implying that every majority quorum of \(A^p\) contains at least one honest validator.

\noindent\textbf{Goal of the set agreement protocol.}
The goal of the set agreement protocol is to output a common recovery set $S^*$ satisfying two properties. First, $S^*$ must contain the extended recovery input of at least one honest validator. Second, all honest validators must agree on the same output $S^*$.
Given such an output, each honest validator $p$ constructs its local recovery DAG $\mathcal{L}_R^p = L_p^{<t_p^*}
\cup
S^*$, after removing duplicate blocks. To see that, if the set $S^*$ satisfies the aforementioned conditions then all validator construct the same DAG, denoted by $\mathcal{L}_R$, observe the following. 

First, every block contained in some ledger $L_p^{<t_p^*}$ of any honest validator $p$ is either already included in the corresponding ledger $L_{p'}^{<t_{p'}^*}$ of every other honest validator $p'$, or belongs to the initial recovery input $S_R^{p'}$ of some honest validator $p'$. In the latter case, $b$ is included in the extended recovery input of all honest validators and, in turn, in $S^*$. Similarly, every block that may have been accepted by a client is either contained  in the ledgers $L_p^{<t_{p}^*}$ of every honest validator $p$ (with a finality certificate), or is included  in the initial recovery set of some honest validator (with a finality certificate or an augmented finality certificate), and, is therefore included in $S^*$.
Consequently, when both conditions are satisfied, all honest validators reconstruct the same recovery DAG $\mathcal{L}_R$, and this DAG contains every block that may have been accepted by a client.


To see how we achieve this objective, we will now describe the exact protocol specification presented in Algorithm~\ref{alg:setagreement}. At a high level, in each view, honest validators send their extended recovery sets to the designated leader, and the leader proposes a set $S$ including all of them. Validators vote only for proposals including their recovery set, ensuring that, whenever a majority quorum certificate is formed for a recovery set $S$, $S$ necessarily contains the extended recovery set of at least one honest validator. To ensure that all honest validators eventually output the same recovery set, the protocol employs a two-phase voting mechanism together with lock certificates, following the approach of~\cite{abraham2020sync}.


\chris{if time allows connect with the pseudocode}
\noindent\textbf{Views and timing.}
The protocol proceeds in views $u=\{1,2,...,n\}$ of duration $8\Delta^*$ each. For validator $p$, view $u$ starts at local time $t_p^u
=
t_p^* + 4\Delta^* + 8(u-1)\Delta^* $. Since the detection time of equivocation, of any two honest validators differs by at most $\Delta^*$, honest validators enter each view within $\Delta^*$ time between each other.





\noindent\textbf{Votes, quorum certificates, and locks.} The protocol uses two types of votes. A \emph{1-phase vote} indicates support for a proposed recovery set, while a \emph{2-phase vote} indicates that the validator is ready to commit the proposal. Both type of votes are signed by the respective validator and also include the proposal of the leader (votes include the hash of proposal of the leader along with the leader signature). \chris{fix here}

Given an active validator set $A$, a \emph{1-phase quorum certificate}, denoted $QC^1(S)$, is a collection of strictly more than $|A|/2$ valid 1-phase votes for recovery set $S$ from validators in $A$.  We also say that $A$ is the set that justifies $QC^1(S)$ (in particular it is, $C_A$, the set of equivocation proofs for all validators in $N \setminus A$ which justifies $QC^1(S)$). Similarly, a \emph{2-phase quorum certificate}, denoted $QC^2(S)$, is a collection of strictly more than $|A|/2$ valid 2-phase votes for $S$ from validators in $A$, and $A$ is the set that justifies $QC^2(S)$.

Each validator $p$ maintains a lock value $QC_l^p$, initially set to $\bot$. A lock is represented as $QC_l^p = (QC^b(S), S, C_A, u)$, where $b \in \{1,2\}$ and $QC^b(S)$ is a quorum certificate for the recovery set $S$, $C_A$ is the set of equivocation proofs defining the active validator set $A$ (i.e., for all validators in $N \setminus A$), and $u$ is the view in which the certificate was formed. Validator $p$ always stores the highest-view lock it has observed. If both a $QC^1(S)$ and a $QC^2(S)$ exist for the same highest view, the validator stores the $QC^2(S)$.


\noindent\textbf{Proposal phase.}
Let $l_u$ be the leader of view $u$. During the interval $[t_p^u, t_p^u+2\Delta^*]$, each validator $p$, sends to the leader $l$ its lock $QC_l^p$ (if not null) along with their recovery set $\hat{S_{R}^{p}}$. 
Since validators enter the delay  with difference of at most  $\Delta^*$ time, the leader $l_u$ collects these messages during $[t_{l_u}^u-\Delta^*,\, t_{l_u}^u+2\Delta^*]$. 

The leader then proposes a tuple $P^u = (QC^u, S^u, u)$, where the lock $QC^u$ and the set $S^u$ are constructed as follows. 






\noindent\emph{Case 1: No valid lock exists.}
If all received lock values are null, the leader proposes a recovery set $S^u = \cup_{ p \in P} \hat{S^p_R}$ where $P$ is the set of all validators that sent their extended recovery set $\hat{S^p_R}$ to the leader, and, thus, $S^u$ is the union of all recovery sets received $[t_{l^u}^u-\Delta^*,\, t_{l^u}^u+2\Delta^*]$. The lock value $QC^u$ in the proposal is null.

Otherwise, if the leader received at least one lock, we have the following cases. \\
\noindent\emph{Case 2a: A unique highest valid lock exists.} If there exists  a unique lock certificate  $QC_l=(QC^{b}(S),S, C_A, u')$ for some view $u'$ that is the highest across all locks collected by the leader $l^u$, then the leader proposes that lock along with the set $S$, i.e., $QC^u = QC_l$ and $ S^u =S$. Also the leader updates its local lock to $ QC^{l_u}_l = QC_l$.  

\noindent\emph{Case 2b: Conflicting  highest valid locks exist.} If the leader observes conflicting lock certificates, $QC_{l_1}$ and $QC_{l_2}$, both with the highest view $u'$ across all locks collected by the leader $l^u$, it sets $QC^u$ as the union of the certificates, and behaves as in Case~1 for the set $S^u$, proposing the union of the collected recovery inputs. 

\noindent\textbf{Voting phase.}
We now describe how validator $p$ chooses whether to vote for a proposal $P^u = (QC^u, S^u, u)$ (beyond verifying the leader signature). In particular, the voting decision depends on whether its lock $QC_l^p$ is null.


\vspace{+0.1cm}

If $QC_l^p = \bot$, then $p$ votes for $P^u$ only if the proposed recovery set $S^u$ \emph{extends} its recovery input $\hat{S_R^{p}}$. Otherwise, if the lock of $p$ is non-null, i.e., $QC_l^p = (QC^b(S), S, C_A, u^*)$, then $p$ votes for $P^u$ only if one of the following holds: 

\noindent\emph{Case 1: Matching locks.} The proposal contains a valid lock certificate $QC^u = (QC^b(S'), S', C_A', u')$ for the same set $S'=S$  and the same view $u' =u^* $.  
\\
\noindent\emph{Case 2: Higher-view lock.} The proposal contains a valid lock certificate $QC^u = (QC^b(S'), S', C_A', u')$ for some view $u' > u^* $, in which case $p$ updates its lock $QC^p_l = QC^u$ and votes for $P^u$ (provided that $C'_A$ correctly defines the active set where $QC^b(S')$ contains strictly more than $|A|/2$ votes). \\
\noindent\emph{Case 3: Conflicting locks.} The proposal $QC^u$ includes at least two lock certificates $QC_{l_1}$ and $QC_{l_2}$ corresponding to different recovery sets $S_{1} \neq S_{2}$, both originating from the same highest view $u' \geq u^*$. In this case, $p$ sets $QC_l^p = \bot$ and behaves similar to the case where $QC_l^p = \bot$ before the proposal $P^u$.

In both cases, if a validator $p$ receives strictly more than $|A^p|/2$ valid 1-phase votes for $P^u$, it forms a 1-phase quorum certificate $QC^1(S^u)$ for the set $S^u$, it updates its lock to $QC_l^p = (QC^1(S^u), S^u, C_{A^p}, u)$ \chris{check if here it's $u$ or the lock's view}, and additionally broadcasts the certificate $QC^1(S^u)$, the set $S^u$, and the proofs $C_{A^p}$ determining its active validator set.

\vspace{+0.1cm}
\noindent\emph{2-phase votes, commit, and termination.} If a $QC^1(S^u)$ for $S^u$ is formed by time $t_p^u + 5\Delta^*$, validator $p$ waits for an additional $2\Delta^*$. If during this period no votes for conflicting proposals are observed (we note that each vote is signed by the leader $l^u$, so Byzantine validators alone cannot trigger this condition), then $p$ broadcasts a 2-phase vote for the proposal $P^u$. If validator $p$ subsequently observes a valid $QC^2(S^u)$ by time $t_p^u + 8\Delta^*$, it commits $S^u$ as the agreed recovery set $S^*$ and broadcasts: (i) the certificate $QC^2(S^*)$, (ii) the set of equivocation proofs $C_{A^p}$ defining the active validator set $A^p$ that justifies $QC^2(S^*)$, and (iii) a commit message for the set $S^*$. 

\chris{strict majority?}
Finally, if the validator $p$ observes commit messages for $H(S^*)$ from a majority of validators in $A^p$, it forms a \emph{commit certificate} for $S^*$ denoted by $CommitQC(S^*)$ and broadcasts the tuple consisting of: (i) $CommitQC(S^*)$,  (ii) the recovery set $S^*$, and (iii) the equivocation proofs $C_{A^p}$ justifying $QC^2(S^*)$. Validator $p$ then outputs $S^*$ and terminates. Any honest validator $p'$ that receives a valid commit certificate for $S^*$ also outputs $S^*$, rebroadcasts the certificate, and terminates.

\begin{algorithm}[t]
\scriptsize
\caption{\textsc{Set-Agreement Protocol} (validator $p$)}
\label{alg:setagreement}
\begin{algorithmic}[1]

\State $A^p,  \hat{S_R^{p}}, QC_l^p \gets \bot$ 

\For{view $u = 1,2,\ldots$}
\State $ vote \gets False$
\If{$l^u$ is leader of view $u$}

    \State $l^u$ collects all $(QC^{p'}_l, \hat{S_R^{p'}})_{p'}$ during $[t_{l^u}^u - \Delta^*, t_{l^u}^u + 2\Delta^*]$ \chris{notation for the set} 
    \State $S^u \gets \bigcup_{p'} \hat{S_R^{p'}}$ \chris{notation for the set} 
    \If{all received locks are null}
        \State $QC^u \gets \bot $ \chris{removing duplicates + set of validators that sent a proposal}
    \ElsIf{a unique highest lock $QC_l'$ exists for some view $u'$ $p \in P$} \chris{notation for the set}
        \State $QC^u \gets QC_l' $
    \ElsIf{conflicting highest locks $QC_{l_1}$ and $QC_{l_2}$ exist} \chris{notation for the set}
        \State $ QC^u \gets (QC_{l_1}, QC_{l_2}) $
    \EndIf

    \State \textbf{broadcast} proposal $P^u=(QC^u, S^u, u)$

\EndIf

\Upon{observing proposal $P^u=(QC^u, S^u, u)$ for the first time}
\If{conflicting valid $QC_1$ and $QC_2$ for view     $u' > u^*$ exist or $QC^u = \bot $} \chris{cross check notation for the views}
    \State{\textbf{if } $S^u \supseteq \hat S_R^p$ \textbf{then} $vote \gets True$}
    \ElsIf{valid $QC^u$ corresponds to a view $ u' > u^* $}
       \State $ QC_l^p \leftarrow QC^u, vote \gets True $
        \EndIf
    \State{\textbf{else if} $QC^u$ corresponds to view $u^* $ and same proposal  $vote \gets True$}
     \State{\textbf{if } $vote$ \textbf{then}
    broadcast $(``1-phase", P^u)$}

\EndUpon

\Upon{observing $QC^1(S^u)$ by time $t^u_p + 5\Delta^*$}
    \State $QC_l^p \gets (QC^1(S^u), S^u, A^p, u )$, \textbf{broadcast} $(QC^1(S^u), C_{A^p})$
    \If{no conflicting proposal observed within the next $2\Delta^*$}
        \State \textbf{broadcast}  $(``2-phase", P^u)$
    \EndIf

\EndUpon

\Upon{observing $QC^2(S^u)$ by time $t_p^u + 8\Delta^*$}
    \State $S^* \gets S^u, QC_l^p \gets (QC^2(S^u), S^u, A^p, u )$
    \State \textbf{broadcast} $(QC^2(S^*), C_{A^p})$
    \State \textbf{broadcast}  $(``commit", H(S^*)$
\EndUpon

\Upon{observing $CommitQC(S^*)$ for $S^*$}
    \State \textbf{broadcast}
    $(CommitQC(S^*), S^*, C_{A^p})$ and \textbf{output} $S^*$

\EndUpon

\EndFor

\end{algorithmic}
\end{algorithm}
\section{Impossibility Proof Sketches}\label{impossibilities:sketch}

Due to space constraints, we present only proof sketches of the impossibility theorems here and defer the formal proofs to 
\ifsubmissionversion
Appendix~F in the extended version~\cite{fullversion}.
\else
Appendix~\ref{sec:impossibilities}.
\fi

Due to space constraints, we present only proof sketches of the impossibility theorems here and defer the formal proofs to 
\ifsubmissionversion
Due to space constraints, we present only proof sketches of the impossibility theorems here and defer the formal proofs to Appendix~F in the extended version~\cite{fullversion}.
\else 
We present only proof sketches of the impossibility theorems here and defer the formal proofs to Appendix~\ref{sec:impossibilities}.
\fi


\vspace{1ex}
\noindent\textbf{\emph{Proof sketch of Theorem~\ref{impossibility:synchrony}.}} 
The adversary can partition correct validators in two distinct subsets $H_1$, $H_2$ and carry out a double-spending attack. 
This attack can be beneficial for rational validators since, for any value of stake $D$ they lock before the run, misbehaving validators can withdraw their stake without even being slashed. 

Fix any stake $D$ and any withdrawal period $\Delta_w$. Consider as input set of conflicting transactions $T_1$ and $T_2$ constructed as explained before. For liveness, correct validators must finalize any input of valid transactions. Thus, validators in $H_1$ must finalize transactions in $T_1$, and validators in $H_2$ must finalize transactions in $T_2$. Since $GST$ occurs at an unknown time, it is impossible for correct validators in different subsets to detect the conflicting blocks including transactions in $T_1, T_2$ before finalizing the transactions. Similarly, it impossible for correct validators to slash the misbehaving ones before they withdraw their stake. 
\qed

\vspace{1ex}
\noindent\textbf{\emph{Proof sketch of Theorem~\ref{impossibility:partialSync2}.}} We modify the attack described in Theorem~\ref{impossibility:synchrony} as follows. 
For any stake value $D$ locked by rational validators before the execution, there exists an input consisting of conflicting transactions $T_1$ and $T_2$ such that rational validators transfer a total of $qD$ coins to clients in $C_1$ and $C_2$, respectively, and hence double-spend  $qD$ coins. Even if the attack is eventually detected and the misbehaving validators are slashed and forfeit their stake, or even if the stake of misbehaving validators is instead offered as a reward for reporting misbehavior, the attack remains strictly beneficial and therefore preferable for rational validators.

Such a configuration can be constructed whenever at least one of the following resources is unbounded: (i) the initial liquid coins held by rational validators, (ii) the coins available as bribes from malicious clients, or (iii) the participation rewards. In particular, when participation rewards are unbounded, rational validators can repeatedly reclaim their initial funds through participation rewards, spend them again in conflicting transactions, and repeat this process sufficiently many times before GST, thereby accumulating a total double-spent value exceeding $qD$.
\qed

\vspace{1ex}
\noindent\textbf{\emph{Proof sketch of Theorem~\ref{impossibility:correctvalidators}.}}
When $f+k \geq q$, Byzantine and rational validators can successfully construct valid certificates for an arbitrary number of conflicting blocks and provide them to clients. For any stake $D$ that rational validators have locked before the run, there is a configuration in which each rational validator can double-spend more than $qD$ coins. That is because the number of clients is arbitrary and unknown before the run. Even if the conflicting blocks are detected and the misbehaving validators get slashed, or even if rational validators can claim the stake of other validators upon reporting an equivocation, the attack is still beneficial for rational validators and would instead prefer to collude with each other. \qed

\vspace{1ex}
\noindent\textbf{\emph{Proof sketch of Theorem~\ref{impossibility:responsiveness}.}} 
Since $q$--commitable protocols are design to operate under partial synchrony and we do not impose any limit in the transaction volume finalized during \( \Delta^* \), the adversary can perform the double-spend attacks described in Theorems~\ref{impossibility:synchrony},~\ref{impossibility:partialSync2}, before the synchronous bound \( \Delta^* \) elapses. Both cases assume the same resilience against Byzantine adversaries and operate under similar network conditions during a network partition in \( \Delta^* \). \qed

\ifsubmissionversion
\else
\section{Analysis} \label{sec:Secanalysis}
\chris{To dos: 1) simultaneously + discuss the "composition" in appendices: i) what happens when we have clients accepting different certificates, ii) we should actually show that the updated protocol still satisfies safety/liveness in partial synchrony} 

We will prove that \sys is both a \emph{safety-preserving} and \emph{liveness-preserving} transformation. To this end, let \( S_f \subseteq S_{\mathsf{bad}} \) denote the set of strategy profiles that cause a safety violation (possibly alongside a liveness violation), and define \( S_l = S_{\mathsf{bad}} \setminus S_f \) as the set of strategy profiles that lead solely to liveness violations.

For each validator \( p \), let \( S_f^p \) (or $S_l^p$) be the set of individual strategies \( s_p \) such that there is a strategy profile \( s = (s_p, s_{-p}) \in S_f \textit{ (or $S_l$)} \) where \( p \) actively  contributes to the safety (or liveness) violation, and let \( S_f^{-p} \) (or $S_l^{-p}$) denote the set of strategies \( s_{-p} \) of all other validators in such profiles.

\subsubsection*{Safety-Preserving Transformation} First, we will show that regardless of the actions of the rest validators it is not beneficial for rational validators to engage in a forking attack. In that way, we demonstrate that for every validator \( p \), and any strategy \( s_p \in S^p_f \), there exists an alternative strategy \( s_p' \in S^p_{good} \) that \emph{weakly dominates} \( s_p \). 





\subsubsection*{{Liveness-preserving transformation}} 





To establish that \sys is a liveness-preserving transformation, we show that at any round there will be future ``good'' leader windows during which at least one correct validator successfully proposes a block. Since this prevents liveness violations, we conclude that for every validator \( p \) and any strategy \( s_p \in S^p_l \), there exists an alternative strategy \( s_p' \in S^p_{\mathsf{good}} \) that \emph{weakly dominates} \( s_p \).




Combining both guarantees demonstrates that the composition of $\Pi$ and \sys forms a coalition compliant protocol. To formalize this, we define the utility function of the rational validators.

\subsubsection*{\textbf{Utility of rational validators}}\label{utility:appendix}
To model dynamic behavior over intervals of rounds, we define the global history up to round \( r \) as \( H^r \), and the local view of player \( i \) at that round as \( H_i^r \subseteq H^r \). Let \( S_i^I \) be the set of strategies available to player \( i \) over an interval \( I \), where each strategy is a function of their local view \( H_i^r \) at each round \( r \in I \). The overall strategy space for the interval is \( S^I = \prod_{i \in [n]} S_i^I \), and the utility function of player \( i \) over \( I \) is given by \( u_i^I : S^I \to \mathbb{R} \).

We now consider the set of strategies available to a rational validator \( p \) over some interval \( I \). First, \( p \) may choose to participate in a coalition that attempts to fork the protocol.\footnote{We assume a worst-case scenario: if \( p \) chooses to join a coalition \( \mathcal{F} \), all other validators in \( \mathcal{F} \) are colluding.} Otherwise, if \( p \) does not engage in a forking attack, it may either follow actions that preserve liveness or abstain from doing so.


\emph{Forking.} To perform this attack, a rational validator $p$ can join a coalition $\mathcal{F} \in \mathcal{F}_R$, where $\mathcal{F}_R$ consists of all the susbets of Byzantine and rational validators sufficienlty large to create a fork. We denote the partipication of $p$ in a coalition $\mathcal{F}$ during the interval $I$ with the indicator variable $b^I_{ \mathcal{F}, p}$. The function $\mathcal{M}$ gives an upper bound of the coins that can be double-spent during the attack: $$\mathcal{M}(\mathcal{F}, I) = \frac{\gamma \cdot C(I)}{R(\mathcal{F})},$$ where $\gamma$ is the maximum number of forks that can be created by the coalition, $C(I)$ is an upper bound of coins finalized per fork during $I$ as specified by the protocol, and $R(\mathcal{F})$ returns the number of rational validators in the coalition\footnote{To provide stronger guarantees, we assume that all the rewards from double-spending are distributed exclusively and equally to the rational validators.}.

Moreover, we denote the participation of $p$ in at least one coalition $\mathcal{F} \in \mathcal{F}_R$ by the indicator variable $b^I_{p}$. Let $ \mathbf{b}^I_{p}$ be the vector whose elements are all indicator variables $b^I_{\mathcal{F},p}$ for every 
$\mathcal{F} \in \mathcal{F}_R$ 
and let $P$ be the slashing penalty imposed on $p$ (which is equal to the stake $D$). 



\emph{Liveness-preserving actions.} 
If validator \( p \) does not join any coalition, it may instead follow the protocol's liveness-preserving actions—such as proposing valid blocks when selected as leader, voting for valid blocks, participating in view-change procedures, and issuing finality votes for blocks committed to its local ledger.

We denote by \( \sigma^I_p \) the strategy of validator \( p \) over the interval \( I \) with respect to these liveness-preserving actions. Let \( R(\sigma^I_p) \) denote the total participation rewards that \( p \) earns by following \( \sigma^I_p \), and let \( B(\sigma^I_p) \) denote any bribes received from an adversary to refrain from signing specific blocks.

Finally, the utility of $p$ in the interval $I$ 
while following the strategy $s_p=(\mathbf{b}^I_{p}, \sigma^I_{p})$  is expressed as:  
\[
u_p^I(s_p, s_{-p}) =
\begin{cases} \sum_{ \mathcal{F} \in \mathcal{F}_R}
    b^I_{\mathcal{F}, p}\cdot\mathcal{M}(\mathcal{F, I}) - P & \text{if $b^I_{p} = 1$}, \\
    R(\sigma^I_{p}) + B(\sigma^I_{p})& \text{otherwise}.
\end{cases}
\]  

\subsection{Lemmas for proving safety}

In Lemma~\ref{proofs:benefitOfFork1}, we prove that once a correct validator finalizes a \( 2\Delta^* \)-old block \( b \), no conflicting block exists. Thus, conflicting block can exist only for \( 2\Delta^* \)-recent blocks. We will use this lemma to establish the recovery guarantees of $\sys$. 
\begin{lemma}\label{proofs:benefitOfFork1}
Assume a correct validator finalizes a block $b_1$ that is $2\Delta^*$-old. Then, no correct validator will ever finalize a block $b_2$ conflicting to $b_1$.

\end{lemma}

\begin{proof}
Assume, for the sake of contradiction, that the following scenario occurs. A correct validator \( p_i \) includes block \( b_1 \) in its local ledger at time \( localTime_{b_1} \) and triggers \( b_1.\mathsf{finalize} \) at time \( localTime \), where \( localTime \geq localTime_{b_1} + 2\Delta^* \). Furthermore, another correct validator \( p_j \) triggers \( b_2.\mathsf{finalize} \) for a block \( b_2 \) that conflicts with \( b_1 \).

By assumption, at \( localTime_{b_1} \), \( p_i \) broadcasts the valid certificate \( QC(b_1) \) for block \( b_1 \). Thus, \( p_j \) receives \( QC(b_1) \) at some time \( t^* \in \{localTime_{b_1}, \dots, localTime_{b_1} + \Delta^*\} \). We distinguish two cases:

  \emph{Case 1:} \( p_j \) triggers \( b_2.\mathsf{finalize} \) at time \( t < t^* \). Then \( p_j \) must have received \( QC(b_2) \) at some time \( t' \leq t < t^* \). Since \( p_j \) broadcasts \( QC(b_2) \), validator \( p_i \) receives it by some time \( localTime_{b_2} \leq t + \Delta^* < localTime_{b_1} + 2\Delta^* \). At time \( localTime_{b_2} \), \( p_i \) detects the equivocation and slashes the misbehaving validators (line~\ref{alg:line:Slash}), resulting in a non-empty blacklist: \( QC(b_1) \cap QC(b_2) \neq \emptyset \). Therefore, at time \( localTime \), the finalization condition (line~\ref{alg:line:deltafinal}) is not satisfied, and \( p_i \) will not finalize \( b_1 \) (line~\ref{alg:line:finalizeD}), contradicting the assumption.

   \emph{Case 2:} \( p_j \) triggers \( b_2.\mathsf{finalize} \) at time \( t \geq t^* \).  At time $t$ \( p_j \) has received both \( QC(b_1) \) and \( QC(b_2) \). Let \( t_{b_2} \) and \( t_{b_1} \) be the times at which \( p_j \) received \( QC(b_2) \) and \( QC(b_1) \), respectively, and assume without loss of generality that \( t_{b_1} \geq t_{b_2} \). At time \( t_{b_1} \), \( p_j \) detects the equivocation and slashes the misbehaving validators (line~\ref{alg:line:Slash}), resulting in \( QC(b_1) \cap QC(b_2) \neq \emptyset \). Therefore, at time \( t \), the finalization condition (line~\ref{alg:line:deltafinal}) is not satisfied, and \( p_j \) will not finalize \( b_2 \) (line~\ref{alg:line:finalizeD}), leading to a contradiction.

\end{proof}

In Lemma~\ref{lemma:bound1}, we show that when validators employ the stake-bounded finalization rule, any rational validator engaging in a forking attack can gain at most \( D \) coins. 
\begin{lemma}  
\label{lemma:bound1}
    Assume that validators employ the stake-bounded finalization rule and that at least one coalition of validators forks the system in some round. 
    \begin{enumerate}
       \item \emph{Upper bound of double-spending}: 
      The coalition double spends up to $fD$ coins and any rational validator participating in the attack can double-spend at most an amount of coins equal to the stake $D$.
        \item {Tight example}: There is a configuration where a rational validator in the coalition double-spends exactly $D$ coins.
    \end{enumerate}
\end{lemma} 
\begin{proof}
\emph{(1) Upper bound of double-spending.} Misbehaving validators can double-spend only coins in conflicting blocks that are accepted by clients. Any inclusion proof $\pi_{tx}$ for a transaction $tx$ included in a block $b$ is accepted by some client $c$ only after witnessing a finality certificate $\mathcal{F}(b)$ including at least $2f+1$ finality votes for $b$. Therefore at least one correct validator $p_i$ has finalized block $b$. Since correct validators do not trigger $finalize$ for conflicting blocks, and at least one correct validator is required to finalize a block, there can be at most $f+1$ conflicting ledgers. 

Consider the time $t$\footnote{To refer at a specific time, we can fix the local clock of any validators, as only the time elapsed matters.} where the first conflicting blocks with valid certificates are formed. Moreover, assume the worst case, where there are $f+1$ conflicting ledgers and let us fix any rational validator $p$ participating in $f+1-m$ conflicting ledgers $L_{i \in [f+1-m]}$ with $m \in [f+1]$. 
The coins that misbehaving validators double-spend on these conflicting ledgers is upper bounded by the total coins spent in conflicting blocks finalized by at least one correct validator.

To bound this, we fix $L_{j}$ to be the longest ledger (i.e., the ledger with the most finalized blocks) at time $t$; if all ledgers have the same length, we fix $L_{j}$
for any $j \in [f+1-m]$. At time $t+ \Delta^*$ every honest validator will have witnessed the equivocation and stop voting for new blocks. To upper bound the double  spent coins during   $[t,t+\Delta^*]$ we fix some $L_{i}$ for any $i \in [f+1-m], i \neq j$ and upper bound the coins that finalized by at least one correct validator in $L_{i}$.

Let $Suffix^{[t, t+\Delta^*]}_{L_i}$ be the set of all blocks with valid certificates in $L_i$ during $[t, t+\Delta^*]$. For every block $b \in Suffix^{[t, t+\Delta^*]}_{L_i}$, $p_i$ updates the variable $suffixvalue$ with the total output sum of transactions in the block $ b$ and only finalizes $b$ if $suffixvalue$ is at most $C=D$ (line~\ref{alg:line:collateralfinal}). Therefore, the amount of coins spent in 
$L_{i}$ is at most $D$.  




Now, we will upper bound the total amount of coins validator $p$ can double-spend in all the $f+1-m$ conflicting ledgers. First, at most $m+1$ correct validators participate in each ledger. We prove that by reaching a contradiction. Assume there is a ledger with at least $m+2$ correct validators. Then, in the remaining $f-m$ ledgers, there must be at least one correct validator. Therefore, there are, in total, at least $m+2+f-m=f+2$ correct validators in the system. We reached a contradiction. Since there are at most $f$ byzantine and $m+1$ correct validators per ledger, at least $f - m $ rational validators participate in each ledger. Therefore, at least $f-m$ rational validators can double-spend at most $C = D$ coins in $f-m$ ledgers, and since they distribute the benefit equally, each validator double-spends $max_{m \in [f+1] }\frac{f-m}{f-m}D= D$ coins. 







\emph{(2) Tight example.}
Consider a run with $f$ byzantine, $f$ rational, $f+1$ correct validators. Moreover, consider as input a set of transactions $Tx = \cup_{i=1}^{f+1} Tx_i$ constructed as follows. There is a partition of clients in distinct sets $ C = \cup_{i=1}^{f+1} C_i$ s.t. for each $i$, in
$Tx_{i}$, each of the $f$ rational validators transfers  $D$ coins to a subset of clients in $C_{i}$. 

At some time, a malicious (Byzantine or rational) validator $p$ is the block proposer. Within the last $\Delta^*$, the adversary partitions the $f+1$ correct validators in the distinct sets $ H_i$ for $i \in [f+1]$.  Now, $p$ constructs the blocks $b_{i \in [f+1]}$, where each $b_i$ includes the transactions in $Tx_i$. Validator $p$ sends the block $b_i$, only to the correct node in $H_i$ and successfully generates a valid certificate for $b_i$ along with the respective finality votes which will present to the clients in $C_i$. Therefore, each rational validator spends the $D$ coins in the $f+1$ conflicting ledgers. In total, each rational validator
double-spends $\frac{f*D}{f} = D$ coins.

\end{proof}

\com{
\textbf{\emph{$i^*$-Strong ledger}:} Consider the correct validator $p_i$, the data structure $ledger$ according to $p_i's$ view, and the maximal subset of blocks $\mathcal{B}$, such that, $ \forall b \in \mathcal{B}, b \in ledger$ and  $b.localtime < localtime - \Delta^*$. Now, let us denote the set of validators that have signed at least one block in $\mathcal{B}$ by $\mathcal{S}_{ledger}$. Since, every block in $ledger$ has at least $2f+1$ finally attestations, $|\mathcal{S}_{ledger}| = 2f + 1 +i^*, i^* \in \{0,1, ..., f\}$. We say that $ledger$ is an $i^*$-strong ledger.

$p_i$ can now finalize a subset of blocks with benefit of forking of at most $mC$ (ref to the updated algorithm). Equation~\ref{eq:strongblocks} outlines how $m$ is calculated. Moreover, in Theorem~\ref{proofs:StrongBlocksSafety}, we show that $\sys$ satisfies Rational Consensus under this finalization rule. 
\chris{should we put $ m \rightarrow \infty$ here too?}
    \label{eq:strongblocks}
   \[  m =   
     \begin{cases}
       \text{1,} &\quad i^* \in \{0, ..., \frac{f}{4}-1\} \\ 
       \text{$\frac{f}{(f-i^*)}C$,} &\quad i^* \in \{ \frac{f}{4}, ..., f\} \\
     \end{cases}
      \]

In Algorithm~\ref{algo:strongcerts} we outline the altered finalization rule given the strong certificates.

Now we examine ...
\emph{strong blocks finalization rule}, who...

}

Now, we argue that \emph{any misbehaving validator will be detected and forfeit their stake}. We emphasize that we consider forks with conflicting finality certificates—not merely valid quorum certificates—as only these can be double-spent and yield a profit for rational participants.

\chris{refer to pseudocode}
\begin{lemma} \label{lemma:successfulslashing} 
    Assume a coalition of validators $\mathcal{F}$ forks the system. No validator in $\mathcal{F}$ can withdraw its stake.
\end{lemma}
\begin{proof}
Assume that the validators in $\mathcal{F}$ fork the system at time \( t \)\footnote{We refer to the first round when distinct correct validators observe conflicting blocks with valid certificates. As before, we use the local clock of any validator to refer to time elapsed between events.} resulting in \( j \in [f+1] \) conflicting ledgers \( L_{i \in \{1, ..., j\}} \) with finality certificates. Let us fix any validator \( p \in \mathcal{F} \). We will show that \( p \) cannot withdraw its stake in any of the conflicting ledgers.

Consider the first ledger \( L_i \) in which validator \( p \)'s withdrawal request \( R_p \) is included in a block \( b \) at time \( t' \). We argue that \( t \leq t' + \Delta^* \). Suppose, for contradiction, that \( t > t' + \Delta^* \). Since block \( b \) has a valid certificate, it must have been witnessed by at least one correct validator \( p_j \), who then broadcasts \( R_p \) at time \( t' \). Therefore all correct validators receive \( R_p \) by time \( t' + \Delta^* \) and ignore the participation of \( p \) at time \( t \). This contradicts the assumption that \( p \) participated in the forking attack at time \( t \).

To complete the withdrawal, a second transaction \( R_{p,s} \) must be committed to \( L_i \) at a time \( > t' + 2\Delta^* \). Since \( L_i \) is the first ledger to include \( R_p \), the same condition applies to any other conflicting ledger, namely the withdrawal can only be finalized after \( > t' + 2\Delta^* \). However, since correct validators broadcast valid certificates, every correct validator will witness all conflicting ledgers at some time \( t^* \leq t + \Delta^* \leq t' + 2\Delta^*\). As a result, they will detect equivocation and blacklist \( p \), preventing the finalization of \( R_{p,s} \). Thus, \( p \)'s withdrawal cannot succeed in any of the conflicting ledgers.



\end{proof}


\begin{lemma}  
\label{lemma:collateralStrongBlocks}
Assume that correct validators employ the \emph{Strong Chain finalization rule} and at least one coalition of validators forks the system. The coalition double spends up to $fD$ coins and any rational validator participating in the attack can double spend at most $D$ coins.
\end{lemma} 
\begin{proof}

Similar to Lemma~\ref{lemma:bound1}, we will fix the longest ledger at the time of the equivocation and upper bound the total amount of coins in conflicting blocks that are finalized by at least one validator during $\Delta^*$.  

\emph{Case 1: Every correct validator finalizes up to $D$.} The proof is similar to Lemma~\ref{lemma:bound1}.

\emph{Case 2: At least one correct validator $p_i$ monitors an $i$-strong chain $L_{p_i}$ with $i > f/4$.} 
Since $1+i$ correct validators monitor $L_{p_i}$, there can be at most $f-i$ additional conflicting ledgers.  We assume the worst case, which is that $f-i$ conflicting ledgers exist. We fix any rational validator $p$ engaging in the attack, participating in $ f-m-i$ conflicting ledgers $L_{i \in [f-m-i]}$, with $m \in [f-i]$. 



For each $i \in [f-m-i]$, at most $m+1$ correct validators and therefore at least $f-m$ rational validators monitor ledger $L_{i}$. We prove that by contradiction. Assume there is a ledger $L_i$ with $m+2$ correct validators. At least $i+1$ correct validators monitor $L_{p_i}$, at least $m+2$ correct validators monitor $L_i$ and at least one correct validator exists in the rest $f-m-i-1$ conflicting ledgers $L_{j\neq i}$. In total, this accounts for $f+2$ correct validators, leading to a contradiction. Furthermore, since there are at most $m+1$ correct validators per ledger, the total amount of coins finalized by correct validators on any ledger is at most $\frac{f}{f-m}D$.


Finally, there are $f-m-i$ conflicting ledgers with at most $f-m$ rational validators spending up to $\frac{f}{f-m}D$ per ledger. So, the amount of double-spent coins is bounded by $max_{\{ m \in [f-i] \}}\frac{(f-m-i)f}{(f-m)^2}D $ which is bounded by $D$ for $i \in \{\frac{f}{4}, ..., f\}$. We can show that, by taking the continuous extension of the function $F(m) = \frac{(f-m-i)f}{(f-m)^2}D$ for $ m  \in [0, f-i]$ and some constant $i \in \{ \lfloor\frac{f}{4}+1\rfloor, ..., f\}$, which takes its maximum value for $m=f-2i$. Since the function takes maximum value at an integer value, the maximum will be the same for the discrete function and thus $\frac{f}{4i} D \leq D$. 

\end{proof}

\begin{lemma}\label{lemma:uniquenessOfStrongestChain}  
At any round, there can only exist a unique strongest chain.
\end{lemma}
\begin{proof}
For the sake of contradiction, assume that at some round $r$, there exist two conflicting chains $T_1$, $T_2$ that both form a strongest chain. That means that $T_1$ is an $i_1$-strong chain for some $i_1 > \frac{f-1}{2}$ and $T_2$ is an $i_2$-strong chain for some $i_2 > \frac{f-1}{2}$. Therefore, there exist two distinct subsets $H_1$ and $H_2$ that consist of $1 + i_1$ correct validators participating in $T_1$ and  $1 + i_2$ correct validators participating in $T_2$ respectively. However, $1+i_1+1+i_2 > f+1$. Since there are $f+1$ only correct validators in the system, we reached a contradiction.
\end{proof}

\begin{lemma}  
\label{lemma:collateralStrongestBlock} 
Assume that correct validators employ the \emph{Strongest Chain finalization rule} and at least one coalition of validators forks the system. The coalition double spends up to $fD$ coins and any rational validator participating in the attack can double-spend at most $D$ coins.
\end{lemma} 
\begin{proof}
We will prove that when validators employ the strongest chain finalization rule, rational validators participating in coalitions can double-spend up to $D$ coins. Similar to Lemma~\ref{lemma:bound1}, we fix the longest ledger at the time of the equivocation and provide an upper bound on the total amount of coins included in conflicting blocks that are finalized by at least one validator within the interval \(\Delta^*\).

\emph{Case 1: There is no correct validator $p_i$ monitoring an $i_1$-strong chain with $i_1 > (f+1)/2$.}
If every correct validator finalizes up to $D$, the proof follows from Lemma~\ref{lemma:bound1}.
If at least one correct validator monitors an $i_1$-strong chain with $i_1 > f/4$, the proof follows from Lemma~\ref{lemma:collateralStrongBlocks}.



\emph{Case 2: At least one correct validator $p_i$ monitors an $i_1$-strong chain with $i_1 > (f+1)/2$.} According to Lemma~\ref{lemma:uniquenessOfStrongestChain}, only one strongest chain can exist. Thus we have the following cases.



    \emph{Case 2a: Every correct validator in a conflicting ledger, finalizes up to $D$.}
    A number of $1+i_1$ correct validators monitor $L_{p_i}$, and thus there are up to $f-i_1$ extra conflicting ledgers.  Assume that there are $f-i_1$ extra conflicting ledgers and let us any rational validator participating in $ f-m-i_1$ of the conflicting ledgers for $m \in [f-i]$. In each fork, there are at most $m+1$ correct validators and at least $f-m$ rational validators (with the same analysis done in Lemma~\ref{lemma:bound1}). Therefore, rational misbehaving validators can double-spend up to $max_{\{ m \in [f-i_1] \}}\frac{(f-m-i_1)}{(f-m)}D \leq  { m \in [f-i_1] \}}\frac{(f-m)}{(f-m)}D = D$.

    \emph{Case 2b: At least one correct validator $p_j$ monitors an $i_2$-strong chain $L_{p_j}$ with $i_2 > f/4$.}. Let $i_1^* = i_1 + 1$ and $i_2^* = i_2 + 1$. In that scenario, at least $i_1^* > \ceil{\frac{f+3}{2}}$ correct validators monitor $L_{p_i}$ and at least $i_2^*$ correct validators, where $ \ceil{\frac{f+4}{4}} < i_2^* < \ceil{\frac{f-1}{2}}$, monitor ledger $L_{p_j}$. Therefore, there are at $ i_3 = f+1- i_1^* - i_2^* < \frac{f-6}{4} <\frac{f}{4} $ remaining correct validators to participate in the rest conflicting ledgers. Thus,  the there are up to $ i_3$ conflicting ledgers, where correct validators finalize up to $D$ coins (since there is a participation of less than $f/4$ correct validators). Assume the worst case, where there are $i_3$ such conflicting ledgers $L_{i \in [i_3]}$. 
    
    Now, we fix any rational validator $p$ participating in $L_{p_i}$ and in $ i_3-m$ of the conflicting ledgers $L_{i \in [i_3-m]}$ for $ m  \in [i_3]$. We will upper bound the coins that $p$ can spend in these conflicting blocks.  First, in $L_{p_j}$, there are finalized up to $\frac{f}{(f-m)(f-i_2)}D <  A(m) = \frac{f}{(f-m)(f-\frac{f-1}{2})}D$. For the conflicting ledgers $L_{i \in [i_3-m]}$, we can show that there are at most $m+1$ correct validators and at least $f-m$ rational validators per fork, with the same analysis done in~\ref{lemma:bound1}. 
    Thus, each validator can double-spend up to $ \frac{f+1- i_1^* - i_2^* - m}{f-m}D = (1+ \frac{1- i_1^* - i_2^*}{f-m})D \leq  B(m) = 1  - \frac{3(f + 2)}{4(f - m)})D$. Now we need to upper bound $C(m) = A(m) + B(m) = (1- \frac{3f^2 + f + 6}{4(f - m)(f + 1)} )D < D$ since 
 $ m  \in [f+1- i_1^* - i_2^*] $ and $f>0$.


\end{proof}

Finally, in Lemmas~\ref{lemma:nonPosUtility},~\ref{lemma:neverFork}, we conclude that rational validators have no incentive to engage in forking attacks, as any potential gain does not exceed their staked amount, which is subject to slashing. 

\begin{lemma}\label{lemma:nonPosUtility} 
When correct validators use the i) Stake-bounded finalization rule, or ii) the Strong Chain finalization rule, or iii) the Strongest Chain finalization rule, for any rational validator, following the protocol weakly dominates participating in a coalition to fork the system. 
\end{lemma}
\begin{proof}

Let us fix any rational validator $p$ and compute its utility to join coalitions forking the system. 
In Lemmas~\ref{lemma:bound1},~\ref{lemma:collateralStrongBlocks},~\ref{lemma:collateralStrongestBlock}, we showed that by performing the attack, $p$ can double-spend at most $D$ coins when validators employ i) finalization rule presented in \ref{sec:protocol}, ii) strong ledger finalization rule, iii) strongest finalization rule respectively. 
According to Lemma~\ref{lemma:successfulslashing}, $p_j$ cannot withdraw its stake $D$ before the conflicting ledgers are detected. Moreover, upon detecting equivocation the system halts and thus misbehaving validators cannot utilize further bribes. 
Therefore, for any strategy $s_p \in S^p_{\mathsf{f}}$ and $s_{-p} \in S^{-p}_{\mathsf{f}}$, the utility of $p$ is
$u_p(s_p,s_{-p}) \leq 0 $, while following the protocol yields non-negative utility. 

To show that following the protocol weakly dominates the strategy of forking the system, we will construct a configuration where $p$ chooses a strategy $s_p \in S^p_{\mathsf{f}}$ and fix the rest strategies $s_{-p} \in S^{-p}_{\mathsf{f}}$ such that $ u_p(s_p, s_{-p}) <0$ (this is sufficient since following the protocol has non-negative utility). To this end, suppose that $p$ participates in a coalition $\mathcal{F}$ consisting of $f$ Byzantine and $\lceil \frac{f+1}{2} \rceil$ rational validators. Moreover, consider a subset $H_1$ of $\lceil \frac{f+1}{2} \rceil $ correct validators, and a subset $H_2$ distinct to $H_1$ that consists of $\lfloor \frac{f+1}{2} \rfloor $ correct validators. Validators in $\mathcal{F}$ construct to conflicting ledgers $T_1$ and $T_2$. 
The utility of $p$ in that case is $u_p(s_p, s_{-p}) \leq \frac{D}{\lceil \frac{f+1}{2} \rceil} - D < 0$.

\end{proof}

\begin{lemma}\label{lemma:neverFork}
    A fork of the system will never occur. 
\end{lemma}
\begin{proof}
    For the sake of contradiction assume that two distinct correct validators $p_1$, $p_2$ finalize the conflicting blocks $b_1$ and $b_2$ respectively. Validator $p_1$ finalized $b_1$ after witnessing the respective valid certificate $QC(b_1)$, and $p_2$ finalized $b_2$ after witnessing $QC(b_2)$. 
By quorum intersection there exists a validator $p$ that is not Byzantine and has voted in both $QC(b_1)$ and $QC(b_2)$. Since correct validators do not equivocate, $p$ must be rational. However, according to Lemma~\ref{lemma:nonPosUtility}, for rational validators, following the protocol weakly dominates participating in a coalition that forks the system. So, $p$ will not sign for conflicting blocks and thus we reached a contradiction.

\end{proof}

\subsection{Lemmas for proving liveness} \chris{extend?}

Since \( \Pi \) is \( (n, f) \)-resilient and \( f + 1 + k \geq n - f \), liveness is ensured as long as rational validators participate in liveness-preserving actions during some intervals. 
In Lemma~\ref{lemma:livenessWeak}, we show that there exist windows during which it is a dominant strategy for rational validators to follow the protocol specifications. Building on this, Lemma~\ref{lemma:liveness} concludes that some correct validator will successfully propose a block that gathers a valid certificate and finality votes.


\chris{consistent notation}
\begin{lemma} \label{lemma:livenessWeak} 
For any round $r$, there is a leader window $W_k$ starting at some round $r' \geq r $, such that following the protocol in every round of $W_k$ is a weakly dominant strategy for any rational validator. \end{lemma}
\begin{proof}

Assume for contradiction that in every window $W_m$ after round $r$, there exists at least one rational validator who prefers a strategy that deviates from the protocol. Since following the protocol yields a non-negative utility, any such deviating strategy must yield strictly positive utility during window $W_m$. Rational validators can deviate as follows:

\noindent\emph{Forking.} From Lemma~\ref{lemma:nonPosUtility}, any strategy in the strategy profile $S^p_f$ (i.e.,  where $p$ actively contributes to a fork), is weakly dominated by following the protocol. Therefore, no rational validator will engage in forking attacks. 
    
\noindent\emph{Only abstaining from liveness actions.} A rational validator will only choose to abstain if the adversary offers a bribe using their participation rewards from previous rounds that outweighs the utility from honest participation.

Suppose the adversary is able to bribe at least one rational validator in every window $W_m$ following round $r$. In such windows, at most $2f$ validators can successfully propose valid blocks. Since participation rewards are distributed only when at least $2f{+}1$ distinct validators succeed, no rewards are paid out in these bribed windows. Because the adversary’s \emph{budget} is finite, it is eventually depleted and the adversary cannot bribe any rational validator anymore. Therefore, there exists a future window $W_k$ (starting at some round $r' \geq r$) such that, for any validator $p$, any strategy $s_p \in S_p$, and any $s_{-p} \in S_{-p}$ under which $p$ abstains from liveness actions during $W_k$, we have $
u_p(s_p, s_{-p}) \leq 0$.

To show that following the protocol during $W_k$ is a weakly dominant strategy, fix any rational validator $p$, and consider two strategies: $s_1$, where $p$ follows the protocol during $W_k$, and $s_2$, where $p$ abstains from liveness actions during $W_k$ (we have already established the result for forking in~\ref{lemma:nonPosUtility}). Let $s_{-p}$ be a strategy profile in which at least $2f$ other validators follow the protocol during $W_k$. In the combined strategy profile $s_a = (s_1, s_{-p})$, at least $2f+1$ validators follow the protocol, satisfying the reward condition, and $p$ receives the corresponding participation rewards (observe that for rational proposers, censoring votes in their proposals only depletes their utility because of the factor $a_k$ so they will include the votes from $p$ in their blocks). In contrast, under $s_b = (s_2, s_{-p})$, only $2f$ validators follow the protocol, so no rewards are distributed. Hence, $u^{W_k}_p(s_a) > u^{W_k}_p(s_b)$.




\end{proof}

\begin{lemma} \label{lemma:liveness} For any round \( r \), there exists a round \( r^* \geq r \) at which a correct validator proposes a block \( b \) such that:  
\begin{enumerate}
    \item \( b \) is finalized by every correct validator, and  
    \item \( 2f+1 \) finality votes for \( b \) or a subsequent block are included in blocks with valid certificates.  
\end{enumerate}
\end{lemma}
\begin{proof}

According to Lemma~\ref{lemma:livenessWeak}, for any round $r$, there exists a round $r' \geq r$ such that a time window $W_k$ begins, during which it is a weakly dominant strategy for rational validators to follow the protocol. Consequently, some block $b$ proposed by an honest validator during $W_k$ will obtain a valid quorum certificate $QC(b)$ and receive at least $2f + 1$ finality votes (either for $b$ or for one of its descendants).


\end{proof}

\com{
\begin{lemma} \label{lemma:finalityVotes}
    For any finalized block $b$, every correct validator can construct a proof $\pi_f(b)$ including the $2f+1$ finality votes for $b$.
\end{lemma}
\begin{proof}
Consider any finalized block $b$ in the view of a correct validator $p$. That means that $p$ has witnessed a valid certificate $QC(b)$ for $b$. Lemma~\ref{lemma:broadcast} illustrates that every other correct or rational validator will have witnessed $QC(b)$ within $\Delta^*$. Since a fork will never occur, as shown in Lemma~\ref{lemma:neverFork}, all correct validators will eventually finalize $b$ (at least as an $\Delta^*$-old block). In that scenario, the $f+1$ correct validators, by assumption, and the $k = 2f-f^* > f$ rational validators, according to Lemma~\ref{lemma:livenessWeak}, will commit their finality votes in subsequent proposed blocks. Eventually, there will be a path of blocks $[b_1, b_2, ..., b_l]$ extending $b$ with the respective valid certificates, where the $2f+1$ finality votes are including within the blocks $[b_1, b_2, ..., b_l]$.
\end{proof}
}

\subsection{Lemmas for Recovery}

For all the lemmas in this section, assume that an equivocation occurs at time $t^*$\footnote{The first time at which an honest validator observes two conflicting SMR certificates $QC(b)$ and $QC(b')$.}. Moreover, we denote the local time that validator $p$ observed the equivocation by $t^*_p$. Our goal is to show that the set agreement protocol outputs a unique recovery DAG $\mathcal{L}$ that contains every block that may have been accepted by a client before recovery begins.


We first show that the dissemination phase ensures that the extended recovery input of an honest validator is a superset of all the initial recovery inputs of honest validators.

\begin{lemma}\label{lemma:initialRecoverySet}
The extended recovery input of any honest validator $p$ includes the initial recovery input of all honest validators, i.e., $ \hat S_R^{p} \supseteq S_R^{p'}$ for every honest validator $p'$. 
\end{lemma}
\begin{proof}
Fix any honest validator $p'$. Validator $p'$ updates its initial recovery set $S_R^{p'}$ until time $t_{p'}^* + 2\Delta^* \leq t^* + 3\Delta^*$. At time $t_{p'}^* + 2\Delta^*$, $p'$ broadcasts $S_R^{p'}$ to all validators and, thus, every honest validator $p$ receives this message by time $t_{p'}^* + 3\Delta^* \leq t^* + 4\Delta^*$. Since $p$ updates its extended recovery set $\hat S_R^{p}$ until time $t_p^* + 4\Delta^*$, and $t_p^* \geq t^*$, we have $t_p^* + 4\Delta^* \geq t^* + 4\Delta^*$. Therefore, $p$ delivers $S_R^{p'}$ from any honest validator $p'$ before completing the construction of $\hat S_R^{p}$, and consequently includes $S_R^{p'}$ in $\hat S_R^{p}$.

\end{proof}

We now prove that every block accepted by clients has been observed by all honest validators.
\begin{lemma}\label{lemma:clientAccepted}
For every block $b$ accepted by a client, and every honest validator $p$, either:
(i) the finality certificate $\mathcal{F}(b)$ is included in $p$'s $t_p^*$--old finalized ledger $L_p^{<t_p^*}$, or
(ii) the finality certificate  $\mathcal{F}(b)$ or the augmented finality certificate $\hat{\mathcal{F}}(b)$ is included in $p$'s extended recovery input $\hat S_R^{p}$.
\end{lemma}
\begin{proof}
Consider any block $b$ accepted by a client. Since clients accept only blocks carrying a valid finality certificate $\mathcal{F}(b)$, at least one honest validator $p'$ must have voted for some blocks including finality votes for $b$ and has therefore observed either $\mathcal{F}(b)$ or $\hat{\mathcal{F}}(b)$ before detecting the equivocation at time $t_{p'}^*$.

In the former case, where $p'$ has observed $\mathcal{F}(b)$ before $t_{p'}^* \leq t^* + \Delta^*$, then $p'$ broadcasts $\mathcal{F}(b)$ and every honest validator $p$ delives the certificate by time $t^* + 2\Delta^* \leq t_p^* + 2\Delta^*$. In that case, either some finality certificate for $b$ $\mathcal{F}(b)$ is already included in $L_p^{<t_p^*}$ (it could even be signed by a different quorum of validators that the finality certificate that $p'$ observes), or $p$ includes $\mathcal{F}(b)$ in its initial recovery input $S_R^{p}$, and hence to $\hat S_R^{p}$ (since $p$ does not delete any certificate from $S_R^{p}$ to $\hat S_R^{p}$). In the latter case, where $\mathcal{F}(b)$ or $\hat{\mathcal{F}}(b)$ is included in $S_R^{p'}$, by Lemma~\ref{lemma:initialRecoverySet}, we have $\hat S_R^{p} \supseteq S_R^{p'}$ for every honest validator $p'$. Therefore, $\mathcal{F}(b)$ or $\hat{\mathcal{F}}(b)$ is included in $\hat S_R^{p}$ of every honest validator $p$.
\end{proof}

We now prove uniqueness of the recovery output. Even though honest validators may collect different recovery inputs, the protocol ensures that at most one set can obtain a valid 2-phase quorum certificate.
\begin{lemma}\label{lemma:RecoverySafety}
At most one proposal set $S^*$ can obtain a valid 2-phase $QC^2(S^*)$ recovery certificate.
\end{lemma}

\begin{proof} 
For the sake of contradiction that assume that two conflicting proposal sets $S_1$ and $S_2 \neq S_1$, both obtain a valid $QC^2(S_1)$ and $QC^2(S_2)$. 

We will first show that this cannot occur if $S_1$ and $S_2$ were proposed in the same view $u$. Since any majority certificate in the active validator set of every honest validator $p$ contains a vote from an honest validator, for both $S_1$ and $S_2$, there exist two honest nodes $p_1$ and $p_2 \neq p_1$ that sign the 2-phase votes for $Q^2(S_1)$ and $Q^2(S_2)$ at times $t_1$ and $t_2$ respectively. Since $p_1$ broadcast a 2-phase vote for $S_1$, at time $t_1 - 2\Delta^*$ node $p_1$ must already have seen the $Q^1(S_1)$ for $S_1$; $Q^1(S_1)$ contains a vote from at least one honest node, say $p_1^*$, and therefore every honest node must have seen $p_1^*$'s vote by time $t \le t_1 - \Delta^*$. Hence, no honest node votes for any conflicting set proposal $S_2$ after time $t_1 - \Delta^*$. Similarly, because $S_2$ gathers a $Q^2(S_2)$, there must exist a $Q^1(S_2)$ certificate including a 1-phase vote from some honest validator $p_2^*$. From the previous argument, $p_2^*$ must have sent the 1-phase vote before time $t_1 - \Delta^*$ (since after this moment it observes the proposal of $S_1$). However, this means that $p_1$ witnesses the 1-phase vote of $p_2^*$ for $S_2$ by time $t_1$ and does not  broadcast a 2-phase vote for $S_1$. We reach a contradiction.

Also, observe that, if some proposal set $S_1$ gathers a $QC^2(S_1)$, only the proposal $S_1$ gathers a 1-$QC$ certificate $QC^1(S_1)$ for view $u$. That is because, no honest validator has provided a 1-phase vote for a proposal conflicting for $S_1$ before $t_1 - \Delta^*$ (else the honest validator $p_1$ would not provide a 2-phase vote for $S_1$), and every honest validator will observe the proposal for $S_1$ by time $t_1 - \Delta^*$ and thus will not broadcast 1-phase phase vote for any conflicting proposal in view $u$ after that time. We will use this observation to prove the statement for different views. 

In particular, assume that two conflicting proposals for sets $S_1$, $S_2$ obtain $QC^2(S_1)$ (including the 2-phase vote of some honest validator  $p_1$) and $QC^2(S_2)$ (including the 2-phase vote of some honest validator  $p_2$). Also, assume that $S_1$ and $S_2$ were proposed in views $u_1, u_2$ respectively, and w.l.o.g, assume that $u_1 < u_2$ and also that $u_1$ is the proposal with the lowest view which has gathered a 2-QC. As explained previously, in view $u_1$ only $S_1$ has gathered a 1-QC certificate. Now, observe that $QC^2(S_1)$ is signed by $p_1$ only if $p$ observes $QC^1(S_1)$ for $S_1$ by local time $ t^u_{p_1} + 5\Delta^*$, in which case it broadcasts $QC^1(S_1)$ ensuring that every honest validator $p_2$ delivers it by time $t^u_{p_2} + 7\Delta^*$ (remember validators start recovery with up to $\Delta^*$ time difference). Since $p_2$ is still in view $u_1$ at this time, it locks on $QC^1(S_1)$, and will never vote for a conflicting proposal, as it is the proposal with highest lock and no conflicting $1-QC$ certificates exist. 



\end{proof}

We now show termination: at least one valid proposal will obtain $QC^2(S^*)$ which is visible to all honest validators.
\begin{lemma}\label{lemma:RecoveryLiveness}
At least one proposal set $S^*$ will obtain a valid 2-phase $QC^2(S^*)$ recovery certificate, and the certificate $QC^2(S^*)$ is visible to all honest validators by time $O(f\Delta^*)$. 
\end{lemma}
\begin{proof}
Let $u$ be the first view whose leader $l^u$ is honest. If an honest validator has already terminated with output $S$ in some earlier view $u' < u$, then all honest validators eventually output the same set $S$ and terminate. This follows because termination requires observing a valid $CommitQC(S)$, which is subsequently broadcast together with the active validator set justifying it. Thus, we focus on the case where no honest validator has terminated before view $u$.  

During the interval $[t_l^u-\Delta^*, t_l^u+2\Delta^*]$, the honest leader $l^u$ collects the locks and extended recovery sets of all honest validators, and (honestly) updates its proposal $P^u$ according to the update lock rule. The proposal is delivered by every honest validator by time $t_l^u+3\Delta^*$. All honest validators then broadcast $1$-phase votes. Therefore, by time $t_l^u+4\Delta^*$, every honest validator $p$ observes a quorum of $1$-phase votes and forms the same $QC^1(S^u)$. Since all honest validators enter the view within at most $\Delta^*$ of each other, each honest validator $p$ delivers $QC^1(S^u)$ by local time $t_p^u+5\Delta^*$, and 
because the leader is honest, only a single proposal is issued in view $u$. Thus, during the waiting period $[t_l^u+4\Delta^*,, t_l^u+6\Delta^*]$, honest validators observe no conflicting $1$-QC and therefore broadcast $2$-phase votes for $S^u$.
All $2$-phase votes are delivered within an additional $\Delta^*$. Hence, by time $t_l^u+7\Delta^* \leq t_p^u+8\Delta^* $, every honest validator $p$ observes at least $|A|/2$ $2$-phase votes, forms $QC^2(S^u)$, and broadcast a commit message for $S^u$. Finally, every honest validator observes a $CommitQC(S^u)$ for $S^u$, outputs $S^u$ and terminates.

\end{proof}

We now argue that every proposal set that obtains a 2-phase quorum certificate must include the initial recovery input of every honest validators.
\begin{lemma}\label{lemma:honestSuperset}
If a set $S^*$ has obtainted a 2-QC $QC^2(S^*)$, then $S^*$ extends the initial recovery input of all honest validators $p$, i.e., $S^* \supseteq S^p_{0}$ for every honest $p$. 
\end{lemma}
\begin{proof}
    Assume that $QC^2(S^*)$ is formed for some proposal $S^*$. The certificate $QC^2(S^*)$ contains votes from a strict majority of some active validator set $A$. After the equivocation is detected, all honest validators identify at least $f+1$ validators that have signed conflicting votes (by quorum intersection).  Since honest validators never equivocate, all honest validators remain in $A$. Therefore, the active set $A$ of any honest validator consists of at most $2f$ validators, among which at least $f+1$ are honest.  This means that, $QC^2(S^*)$  contains at least one vote from some honest validator, which only votes for $QC^2(S^*)$ upon observing a 1-QC $QC^1(S^*)$ for $S^*$. Similarly, this 1-QC is signed by at least some honest validator $p$. Since this honest validator $p$ only votes for proposals which extends their extended recovery set $\hat S_R^{p}$, it holds that $ S^* \supseteq \hat S_R^{p} $, and by lemma~\ref{lemma:initialRecoverySet} it follows that  $ S^* \supseteq S_R^{p'}$ for all honest validators $p'$.


\end{proof}

Finally, we show that once agreement on $S^*$
is reached, honest validators reconstruct the same recovery DAG.
\begin{lemma}\label{lemma:commonRecoveryLedger}
Let $\mathcal{L}^p_R$ be the result of the union of the  $L_p^{<t_p^*} \cup S^*$ after removing duplicates, where
$L_p^{<t_p^*}$ is the $t_p^*$--old finalized ledger of $p$ and $S^*$ is the output of the set agreement protocol. Then, i) every honest validator $p$ obtains the same DAG $\mathcal{L}_R $, ii) every block accepted by clients is included in $\mathcal{L}_R$. 
\end{lemma}

\begin{proof}
Fix any honest validator $p$. Since $S^*$ is the unique output of the set agreement protocol, it is identical for all honest validators (by Lemma~\ref{lemma:RecoverySafety}). It remains to show that every block $b \in L_p^{<t_p^*}$ is also included in $L_{p'}^{<t_{p'}^*} \cup S^*$, of any honest validator $p'$.

To show that, fix any block $b \in L_p^{<t_p^*}$. Then, $p$ has observed a  finality certificate $\mathcal{F}(b)$ for $b$ before time $t_p^*$, and since $p$ broadcasts $\mathcal{F}(b)$ at time $t_p^* \leq t^* + \Delta^*$, every honest validator $p'$ delivers it by time $t^* + 2\Delta^*$. Now, we have the following cases: i) either all validators $p'$ have already observed some finality certificate $\mathcal{F}(b)$ for $b$ before $t_{p'}^*$, in which case $b \in L_{p'}^{<t_{p'}^*}$, or ii) there exist some honest validator $p'$ that includes $\mathcal{F}(b)$ in its initial recovery set $S_R^{p'}$. Since $S^* \supseteq S_R^{p'}$( Lemma~\ref{lemma:honestSuperset}), it follows that $b \in S^*$.

We now show that every block accepted by a client is included in
$\mathcal{L}_R$. Consider any block $b$ accepted by a client. For block $b$, it holds that, either: (a) $\mathcal{F}(b)$ is included in $L_p^{<t_p^*}$ of every honest validator $p$ ; or
(b) $\mathcal{F}(b)$ or $\hat{\mathcal{F}}(b)$ is included in the
initial recovery input $ S_R^p$ of some honest validator $p$. In case (a), $b$ belongs to the finalized ledger component of
$\mathcal{L}_R$ (by construction of $\mathcal{L}_R$). In case (b), the
initial recovery input of every honest validator is including in $S^*$ by 
Lemma~\ref{lemma:honestSuperset}.

\end{proof}

\subsection{Final Theorems}

\begin{theorem}\label{theorem:totalordering}
Validators running the composition of an $(3f+1,f)$-resilient,  $(2f+1)$--commitable SMR protocol $\Pi$ and \(\sys\) agree on a total order in all executions with $f^* \leq f$ and $k \leq 2f - f^*$.
\end{theorem}
\begin{proof}
\emph{Safety.} In Lemma~\ref{lemma:neverFork}, we show that validators will never finalize conflicting blocks. 

\emph{Liveness.} Consider any transaction $tx$ witnessed by all correct validators at some round $r$. According to Lemma~\ref{lemma:liveness}, there is a round $r^*\geq r$ where a correct validator $p$ will successfully propose a valid block $b$. If $tx$ has not already been included in a valid block proposed before round $r^*$, then, $p$ will include $tx$ at block $b$.


\end{proof}

\begin{theorem}\label{theorem:certifiable}
The composition of an $(3f+1,f)$-resilient,  $(2f+1)$--commitable SMR protocol $\Pi$ and \(\sys\)  satisfies certifiability in all executions with $f^* \leq f$ and $k \leq 2f - f^*$.
\end{theorem}
\begin{proof}
\emph{Client-safety.}
Consider any client $c$. We will show that any transaction $tx$ for which $c$ has accepted an inclusion proof is included in the transaction ledger of at least one correct validator. 

To this end, consider the inclusion proof $\pi_{tx}$ for the transaction $tx$. Client $c$ has accepted $\pi_{tx}$ only after receiving a proof showing the inclusion of $tx$ in a block $b$, where $b$ has $2f+1$ finality votes. At least one correct validator $p$ has committed a finality vote for $b$. Validator $p$ commits finality votes only for blocks included in its transaction ledger.

\emph{Client-liveness.} 
For every transaction \( tx \) committed to a block \( b \) accepted by all correct validators, any correct validator \( p_i \) can construct a proof \( \pi \) as follows. Validator \( p_i \) sends an inclusion proof \( \pi_{tx} \) of \( tx \) in \( b \) to a client \( c \), along with a proof \( \pi_f(b) \) consisting of \( 2f+1 \) finality votes for \( b \) or subsequent block, which is feasible according to Lemma~\ref{lemma:liveness}. Any client \( c \) will accept the proof $\pi$ by definition of the client confirmation rule.

\end{proof}

\chris{fix}
\begin{theorem}\label{th:recovery}
The composition of an $(3f+1,f)$-resilient,  $(2f+1)$--commitable SMR protocol $\Pi$ and \(\sys\) regains economic restitution with recovery parameter $\Delta_R = O(f \Delta^*)$ in executions with 
$f^* < 2n/3$.
\end{theorem}
\begin{proof}

Assume an equivocation on the client side occurs, i.e., conflicting blocks containing finality votes are formed, and let time \(t^*\)\footnote{We fix the local clock of any participant, as only elapsed time matters.} be the first time of an equivocation. After witnessing an equivocation, validators agree on a common DAG $\mathcal L_R$ (lemmas~\ref{lemma:RecoveryLiveness},~\ref{lemma:commonRecoveryLedger}). To execute the blocks of $\mathcal L_R$, each validator returns to the state of its last \(2\Delta^*\)-finalized block; observe that by construction of $\mathcal L_R$ the last \(2\Delta^*\)-finalized block of each honest validator is a prefix of $\mathcal L_R$. Moreover, by Lemma~\ref{proofs:benefitOfFork1}, no block conflicting with this block exists.  
Using the same deterministic rule (e.g., ordering by the hash of the last finalized block), all validators execute transactions from the conflicting ledgers in the same order.

Any transaction issued by an honest sender is executed successfully, since the sender holds the corresponding coins in the respective ledger. Transactions whose recipients lack sufficient balance correspond to double-spent coins. By Lemmas~\ref{lemma:bound1},~\ref{lemma:collateralStrongBlocks},~\ref{lemma:collateralStrongestBlock}, during the interval \([t,\, t+\Delta^*]\), at most \(fD\) coins are double-spent in total. Because, all honest validators identify (the same) at least \(f+1\) misbehaving validators, these coins are covered by the stake of misbehaving validators which is in total $(f+1)D > fD$ (ensuring monetary conservation); any remaining stake of misbehaving validators is burned (line~\ref{alg:line:Slash}) guaranteeing accountability. 

Thus, all honest validators compute the same resulting state and sign the same state commitment by time $O(f\Delta^*)$ ensuring deterministic recovery. In particular, at least \(f+1\) honest validators sign the state commitment which they broadcast to every validator (line~\ref{alg:line:broadcastCmt}). Thus every honest validator will observe the verifiable state commitment by time $O(f\Delta^*)$. In the resulting state, all affected clients are reimbursed and all misbehaving validators are penalized.

\end{proof}

\com{
Consider the subset of validators $\mathcal{F_R}$ misbehaving, i.e., forking the system at some time $t$, and assume there are extra $\gamma = f-m-i^*$ forks, with $m \in [f-i^*]$ since there is at least one honest validator per fork while at least $1+i^*$ correct validators are participating at the ledger $p_1$ views. Since $p_i$ finalizes at $\frac{f}{f-i^*}C$, it must be $i^*>\frac{f}{4}$. We will show that the amount of coins that each rational validator has double spent in the rest forks is bounded by $C$.

At most $m+1$ correct validators and therefore at least $f- m$ rational validators participate in each ledger of the fork. We prove that by contradiction. Assume there is a fork with $m+2$ correct validators. Therefore there are in total $i^*+1$ in $C_{p_i}$, $m+2$ with that fork, and at least one correct validator in the rest forks leading to $f+2$ correct validators in total. Moreover, since there are at most $m+1$ correct validators per fork, in every fork, correct validators finalize at most $\frac{f}{f-m}C$ coins per fork (by employing the \emph{strong block finalization rule}).

Therefore, there are $f-m-i^*$ forks with at most $f-m$ rational validators to share utility of at most $\frac{f}{f-m}C$ per fork, so the total utility is bounded by $max_{\{ m \in [f-i^*] \}}\frac{(f-m-i^*)f}{(f-m)^2}C$ which is bounded by  $C$ for $i^* \in \{\frac{f}{4}, ..., f\}$. \chris{I'm not writing the calculations here} $\square$
}

\chris{fix}

\section{Impossibility results} \label{sec:impossibilities}

\chris{fix}
\subsubsection*{Proof of Theorem~\ref{impossibility:synchrony}}
\begin{proof}
For the sake of contradiction assume that there is a q--commitable SMR protocol $\Pi$ that is $(n, k, f)$-resilient for $f\geq \lceil h/2 \rceil $, where $h=n-k-f$ is the number of correct validators. 
Assume that each validator has deposited $D \in \mathbf{N}$ stake to participate in $\Pi$ and assume any period $\Delta_W$ where a validator can withdraw their stake.
 

We will construct a run where colluding with Byzantine validators to create valid certificates for conflicting blocks yields greater utility  for rational validators. 


\emph{Double spending attack.}
Since $ f \geq \lceil h/2 \rceil $, it holds that $ f + k + \lfloor h/2 \rfloor  \geq  q$. In that case, a malicious validator can partition correct validators in two distinct subsets $H_1$ and $H_2$ monitoring two conflicting ledgers $ T_c \cup T_1$ and $ T_c \cup T_2$ respectively, where $T_c$ is the common part of the ledgers. Assume the attack takes place during an interval 
\footnote{We can use a reference point the local clock of any correct validator.}
$[t,t']$ where $t'> t + \Delta_W$  and let $\mathcal{F}$ be the set of the $f+k$ misbehaving validators. 

We will construct $T_1$ and $T_2$ as follows. There is a partition of the clients in two sets $C_1, C_2$, and a set of transactions $Tx = Tx_1 \cup Tx_2$ input to the SMR protocol $\Pi$, where a set $\mathcal{R}$ of at least $q - f - min(|H_1|, |H_2|)$ 
\footnote{If \( q - f - \min(|H_1|, |H_2|) = 0 \), meaning there are enough Byzantine validators to carry out the attack, the analysis follows trivially.}
rational validators spend all their liquid coins $\epsilon >0$ in $Tx_1$ to a subset of clients in $C_1$ and the same $\epsilon>0$ coins in $Tx_2$ to a subset of clients $C_2$. Moreover, in $Tx_1$ and $Tx_2$, there are included all necessary transactions so that validators can withdraw their stake.


We will show that such an attack is feasible. That is, if rational nodes in $\mathcal{R}$ perform the attack, it leads to a safety violation. Consider a correct validator $p_1 \in H_1$ and any correct validator $p_2 \in H_2$. Since $\Pi$ solves SMR, it guarantees liveness, and therefore, every transaction in $Tx$ is committed to the ledger of $p_1$ and $p_2$ 
\footnote{We assume that correct validators gossip valid transactions, so the transaction will eventually be witnessed by every correct validators.}
. We prove that $p_1$ and $p_2$ will finalize the blocks in the conflicting ledgers $T_1$ and  $T_2$ respectively, using an indistinguishability argument.  


\emph{World 1.}
(GST = 0.) Validators in $H_1$ and $F$ are correct. Validators in $H_2$ are Byzantine and do not participate in the protocol. The input to $\Pi$ are the transactions in $Tx_1$. Validator $p_1 \in H_1$ witnesses blocks in $T_1$ 
\footnote{We mean that $p_1$ sequentially finalizes blocks in $T_1$.}
with quorum certificates of $n-f$ votes.

\emph{World 2.} (GST = 0) This is symmetrical to World 1. Now, validators in $H_2$ and $F$ are correct. Validators in $H_1$ are Byzantine and do not participate in the protocol. The input to $\Pi$ are the transactions in $Tx_2$. Validator $p_2 \in H_2$  witnesses blocks in $T_2$ along with the quorum certificates of $n-f$ votes.

\emph{World 3.}
($GST = t'$). Validators in $F$ perform the Double Spending Attack, as explained before, until $GST$. First, the adversary partitions the network between validators in $H_1$ and $H_2$, such that messages sent by validators in $H_1$ are received from validators in $H_2$ only after $GST$ happens, and vice versa. The input to $p_1 \in H_1$ are the transactions in $Tx_1$ and the input to $p_2 \in H_2$ are the transactions in $Tx_2$. Then, every validator in $F$ simulates two correct validators. The first one simulates the execution of World 1 to any validator $p_1 \in H_1$ and the second one simulates the execution of World 2 to any validator $p_2 \in H_2$.

In the view of $p_1$ World 1 and World 3 are indistinguishable. Thus, $p_1$ finalizes blocks in $T_1$ after $T_c$. In the view of $p_2$ World 2 and World 3 are indistinguishable. Therefore, $p_2$ finalizes blocks in $T_2$ after $T_c$ leading to a safety violation.

Then, to argue that the attack is profitable for rational nodes we will show that \emph{clients will accept the inclusion proofs for conflicting blocks}.  Since $\Pi$ satisfies client-liveness, there must be a rule so that clients in $C_1$ (or $C_2$) can receive inclusion proofs for blocks in $T_c$. We will show that any rule $\rho$ that satisfies client-liveness, leads to clients accepting the inclusion proofs for the conflicting blocks in the double spending attack.

In q--commitable protocols, the rule $\rho$ can either depend on some time bound $\Delta$ or on a certificate of votes. However, we show that rule $\rho$ cannot depend on a time bound in the partial synchronous model using an indistinguishability argument. Consider the time bound $\Delta$ after GST. Moreover, consider any rule $\rho_1$ according which clients wait for the message delay $k\Delta$ for any $k \in \mathbb{Z}$ to receive a proof of equivocation  and, w.l.o.g., let us fix some client $c_1 \in C_1$. 

\emph{World A.}
(GST=0). The next blocks with valid certificates extending $T_c$ are the blocks in $T_1$. Validator, $p_1$ commits blocks in $T_1$ in its local ledger.
Client $c_1$ receives a valid certificate for any block $b \in T_1$ from $p_i$ and waits for $k\Delta$. Since there are not conflicting blocks in the system, client $c_1$ does not receive a proof of equivocation.


\emph{World B.} ($GST > t'+ k\Delta)$.
Validators in $\mathcal{F}$ perform the double spending attack at time $t$. Validator $p_1\in H_1$ finalizes blocks in $T_1$ to extend $T_c$ and validator $p_2\in H_2$ finalizes blocks in $T_2$ to extend $T_c$. Validator $p_1$ sends the inclusion proof for blocks in $T_1$ to $c_1$. Client $c_1$ waits for $k\Delta$. At $t' + k\Delta < GST$, $c_1$ does not receive any message from validators in $H_2$ and does not witness conflicting blocks.

Client $c_1$ cannot distinguish between World A and World B and therefore, the rule $\rho_1$ is not sufficient. Next, consides certificates of stages of at least $t$ votes. since the $f$ Byzantine validators might never participate, $t \leq n-f $. Now, consider any rules $\rho_a$ and $\rho_b$, which are receiving stages of at least $t_1 = n-f$ votes and $t_2$ votes for some $t_2 \in [1, n-f)$ respectively. Any certificate validating $\rho_1$ also validates $\rho_2$. Therefore, an inclusion proof consisting of stages of $t_1 = n-f \geq q$ votes is the strongest certificate that clients can receive, and will therefore accept the inclusion proofs for the blocks in the conflicting ledgers $T_1, T_2$.  \footnote{A different approach is to notice that $c_1$ cannot distinguish between World 1 and World 2 of the Double Spending Attack.} 

\end{proof}

\subsubsection*{Proof of Theorem~\ref{impossibility:partialSync2}}\begin{proof}
For the sake of contradiction, assume that there exists a $q$--commitable SMR protocol $\Pi$ that is $(n, k, f)$-resilient for some $q$ s.t. $q \le f + k + \lfloor h/2 \rfloor$, where $h = n - k - f$ is the number of correct validators. Also, assume that each validator has deposited $D \in \mathbf{N}$ stake to participate in $\Pi$.

We construct an attack in which rational validators collude with Byzantine validators to produce valid certificates for conflicting blocks, yielding strictly greater utility for rational validators and leading to a safety violation. 
Similar to the proof of Theorem~\ref{impossibility:synchrony},  a malicious validator  can partition correct validators in two distinct subsets $H_1$ and $H_2$ which monitor the conflicting ledgers $ T_c \cup T_1$ and $ T_c \cup T_2$ respectively. Additionally, there is a partition of the clients in two sets $C_1, C_2$, and the input to the blockchain protocol is $Tx = Tx_1 \cup Tx_2$, where a set $\mathcal{R}$ of at least $p - f - min(|H_1|, |H_2|)$ rational validators spend $qD$ coins in $Tx_1$ to a subset of clients in $C_1$ and $qD$ coins in $Tx_2$ to a subset of clients $C_2$ where $\epsilon_1, \epsilon_2 > 0$. Therefore, even if the attack is eventually detected and rational validators forfeit their stake,
or if the stake of misbehaving validators is instead offered as a reward for reporting misbehavior, the total value double-spent makes the attack strictly beneficial for rational validators.

We will now construct a run where, for any $D$, every rational validator spend $qD$ coins in $T_1$ (same for $T_2$). To this end, cosnider the following notation. $G^p_{0}$ denotes the coins of validator $p$ in $T_c$ plus the maximum amount of bribes from malicious validators to fork the system. Then, we denote by $S^p_i$ the total coins paid to clients in $C_1$ by validator $p$ and by $R^p_i$ the participation rewards $p$ has received until the $i$-th block of $T_1$ (excluding $T_c$).  

We need to show that for every rational validator $p$ there is a $k$ s.t. $S^p_k > qD$, or $ G^p_0 + R^p_k > qD$. If the initial coins held by $p$, or the coins provided by malicious clients as bribes, are unbounded, this condition trivially holds. We now show that the condition also holds if the participation rewards $R^p_k$ are unbounded.

Toward this direction, consider the first $i \in \mathcal{N}$
such that $ R^p_i = \mathcal{C}^p_{0}$ as follows. Every validator in $\mathcal{R}$ spends all their coins, and therefore all coins belong to the hands of all the other users. When enough blocks has been processed in the system, any rational validator $p$ actively participating (by proposing or voting blocks depending on the reward scheme) will receive at least $\mathcal{C}^p_{0}$ in participation rewards. Specifically, w.l.o.g., if the participation rewards per block are $r$,$n$ blocks have been processed in the system and $p$ has participated in a percentage of $\alpha$ of them, then there is a run where $\alpha \cdot n \cdot r \geq \mathcal{C}^p_{0} $, which holds since it can be $n \to \infty$. Then, rational validators will spend all of their money to pay clients in $C_1$ and the above scenario is repeated for $j$ times such that there is $i \in \mathcal{N}$, $ R^p_i =  j \mathcal{C}^p_{0} $  and $(j+1) \mathcal{C}^p_{0} > qD$. Since it can be that $ j \to \infty$, this is a feasible scenario.
\end{proof}


\subsubsection*{Proof of Theorem~\ref{impossibility:correctvalidators}} 
\begin{proof}
For the sake of contradiction assume that there is a q--commitable SMR protocol $\Pi$ that is $(n, k, f)$-resilient for $f + k \geq h$, where $h=n-k-f$ is the number of correct validators. 
Assume that each validator has deposited $D \in \mathbf{N}$ stake to participate in $\Pi$.  


We will construct a run where it is a dominant strategy for rational validators to collude with Byzantine validators and create valid certificates for blocks that are not witnessed, and therefore not finalized, by any correct validator. Clients will accept the certificates, leading to a client-safety violation. Towards this direction, consider a run of $\Pi$ and let us denote by $C_t$ the number of clients in the system at time $t$. 
Assume that all correct validators agree in the transaction ledger $T_c$ at time $t-\epsilon$ for some $\epsilon>0$. 

\emph{Double spending attack.} The input of $\Pi$ at time $t$ is a set of transactions  $Tx = \cup_{i=1}^{C_t} Tx_i$ constructed as follows. There is a set $\mathcal{R}$ of at least $ r = max(q-f, 0)$ rational validators such that every validator in $\mathcal{R}$ transfers $c>0$ coins to the client $c_{i\in\{1,\dots,C_t\}} $ in $Tx_{i\in\{1,\dots,C_t\}}$. Moreover, we construct the run such that $(C_t-1)C > D$. Since the deposit $D$ is decided before the run of the protocol and $C_t$ is arbitrary, it is feasible to construct such a run. 

Validators in $\mathcal{R}$, along with the $f$ Byzantine validators, create the blocks $b_{i\in\{1,\dots,C_t\}}$, where for each $i, b_i$ includes the transactions $Tx_i$. Since $r+f \geq q$, they can construct $s$ phases (for any $s$ specified by $\Pi$) of $q$ votes for every blocks $b_{i\in\{1,\dots,C_t\}}$. Note that none of the correct validators has witnessed the valid certificate for any of those blocks. 

We will show that clients $c_{i\in\{1,\dots,C_t\}} $ will accept the inclusion proofs for blocks $b_{i\in\{1,\dots,C_t\}}$. In that scenario, any rational validator $p \in \mathcal{R}$ double spends $ (C_t-1)C > D$ coins. Even if  validator $p$ is slashed in $\Pi$, the attack is still beneficial. This will lead to violation of client-safety.

\emph{Clients must accept the inclusion proofs.} Since $\Pi$ ensures client-liveness, there is a rule $\rho$ that allows clients to accept inclusion proofs for every block committed to the ledger $T_c$. We will show that for any such rule $\rho$, all clients $c_{i \in \{1, \dots, C_t\}}$ will accept inclusion proofs for blocks $b_{i \in \{1, \dots, C_t\}}$.  

First, consider a rule $\rho_1$ where clients wait for a fixed time bound $k\Delta$, where $\Delta$ is the synchrony bound and $k > 0$. We fix a client $c_{i \in \{1, \dots, C_t\}}$. We prove that this rule is insufficient using an indistinguishability argument.


\emph{World 1.}
The next block extending $T_c$ with a valid certificate is block $b_i$. At least one correct validator $p_i$, commits $b_i$ in its local ledger. There are not conflicting blocks in the system. Client $c_i$ receives a valid certificate for $b_i$ from $p_i$ and waits for $k\Delta$. After the period passes $c_i$ detects no conflict. 

\emph{World 2.} 
Rational nodes collude with Byzantine nodes to perform the double spending attack described above. A rational node $p \in \mathcal{R}$ sends the inclusion proof of $b_i$ to $c_i$. Client $c_i$ waits for $k\Delta$. Since no correct validator has witnessed the certificates for any conflicting block $b_{j\in\{1,\dots,C_t\}, j \neq i}$, client $c_i$ detects no conflict. 

Client $c_i$ cannot distinguish between World 1 and World 2 and therefore, the rule $\rho_1$ is not sufficient. Now, consider any rule $\rho_1$ which requires stages of at least $t$ votes. Similar to the proof of Theorem~\ref{impossibility:synchrony}, an inclusion proof of stages of $t = n-f$ votes is the strongest certificate that clients can receive.  Therefore, each client $c_{i\in\{1,\dots,C_t\}}$ clients will accept the inclusion proofs for the respective block $b_{i}$. 


\end{proof}


\subsubsection*{Proof of Theorem~\ref{impossibility:responsiveness}} 
\begin{proof}
Assume, for contradiction, that there exists a $q$--commitable SMR protocol $\Pi$ that is $(n, k, f)$-resilient  for i) $f \geq \lceil h/2 \rceil$, ii) or $q \leq f + k + \lfloor h/2 \rfloor $, where $h = n - k - f$ represents the number of correct validators, and $\Pi$ does not incur a finality delay of at least $\Delta^*$. Moreover, suppose each validator has staked $D \in \mathbf{N}$ to take part in $\Pi$.   

Theorems~\ref{impossibility:synchrony} and~\ref{impossibility:partialSync2} share the same Byzantine resilience and assume similar network conditions during a partition lasting up to~$\Delta^*$, as in Theorem~\ref{impossibility:responsiveness}, with $f \ge \lceil h/2 \rceil$ and $q \le f + k + \lfloor h/2 \rfloor$, respectively. Before the bound $\Delta^*$ elapses, a $q$--commitable protocol behaves as in partial synchrony.  
Thus, the attacks presented in proofs of Theorems~\ref{impossibility:synchrony} and~\ref{impossibility:partialSync2} can be extended and show that: (1) \emph{safety is violated}, with correct validators $p_1 \in H_1$ and $p_2 \in H_2$ finalizing conflicting blocks in ledgers $T_1$ and $T_2$, and (2) \emph{clients accept inclusion proofs} backed by valid certificates.
\end{proof}
\qed



\fi


\IEEEpeerreviewmaketitle

\end{document}